\RequirePackage{snapshot}
\documentclass[preprint]{imsart}  
\usepackage{amsmath,amssymb,amsthm}
\usepackage[OT1]{fontenc}
\usepackage[utf8]{inputenc}
\usepackage[numbers]{natbib}
\usepackage[colorlinks,citecolor=blue,urlcolor=blue]{hyperref}
\usepackage{color}
\usepackage{indentfirst}
\usepackage{multirow}
\usepackage{float}
\usepackage{subfig}
\usepackage[english]{babel}
\usepackage{bm}
\usepackage{graphicx}

\arxiv{http://arxiv.org/abs/1410.3687}

\usepackage{array}
\newcolumntype{+}{>{\global\let\currentrowstyle\relax}}
\newcolumntype{^}{>{\currentrowstyle}}
\newcommand{\rowstyle}[1]{\gdef\currentrowstyle{#1}%
  #1\ignorespaces
}  

\usepackage{arydshln}

\startlocaldefs
\numberwithin{equation}{section} \theoremstyle{plain}
\newtheorem{theorem}{Theorem}[section]
\newtheorem{lemma}{Lemma}[section]
\newtheorem{corollary}{Corollary}[section]
\newtheorem{proposition}{Proposition}[section]

\newtheorem{remark}{Remark}[section]
\endlocaldefs

\begin{document}
\newcommand\cov{\mathop{\text{cov}}}
\newcommand\tr{\mathop{\text{tr}}}
\newcommand\E{\mathbb{E}}
\newcommand\R{\mathbb{R}}
\newcommand\N{\mathcal{N}}
\newcommand\lb{\left(}

\newcommand\re[1]{{\color{black}#1}}  
\newcommand\bl[1]{#1}
\newcommand\bll[1]{{\color{blue}#1}}
\newcommand\gr[1]{{\color{green}#1}}

\newcommand\rb{\right)}
\newcommand\veps{\varepsilon}
\renewcommand\hat[1]{\widehat{#1}}
\newcommand{\tb}[1]{{#1}}

\begin{frontmatter}
  \title{Identifying the number of factors from singular values of a
    large  sample auto-covariance matrix}
  \runtitle{Identifying number of factors}
  %
  %
  %
  %
  \begin{aug}
    \author{\fnms{Zeng} \snm{Li}\ead[label=e2]{u3001205@hku.hk}}
    \and
    \author{\fnms{Qinwen} \snm{Wang}\ead[label=e1]{wqw8813@gmail.com}}
    \and
    \author{\fnms{Jianfeng} \snm{Yao}\ead[label=e3]{jeffyao@hku.hk}}

    \runauthor{Z. Li, Q. Wang and J. Yao}

    \affiliation{Zhejiang  University and The University of Hong Kong}

    \address{Qinwen Wang  \\
      Department of Mathematics\\
      Zhejiang University \\
      \printead{e1}
    }

    \address{Zeng Li, Jianfeng Yao\\
      Department of Statistics and Actuarial Science\\
      The University of Hong Kong\\
      \printead{e2,e3}
    }
  \end{aug}
  \begin{abstract}
    Identifying the number of factors in a high-dimensional
    factor model has attracted much attention in recent years
    and a general solution to the problem is still lacking.
    A promising ratio estimator based on the singular values of the
    lagged autocovariance matrix has been recently proposed in the
    literature and is shown to have a good performance under
    some specific assumption on the strength of the factors.
    Inspired by this ratio estimator and as a first main contribution,
    this paper proposes a  complete
    theory of such sample singular values for both the factor part and
    the noise part under the large-dimensional scheme where the
    dimension and the sample size proportionally  grow to infinity.
    In particular, we  provide  the exact description of the phase
    transition phenomenon that determines whether a factor is strong
    enough to be detected with the observed sample singular values.
    Based on these findings and as a second main contribution of the
    paper,
    we propose a new
    estimator of the number of factors which is strongly consistent
    for the detection of  all significant  factors
    (which are the only theoretically detectable ones).
    In particular, factors are assumed to have  the minimum strength
    above the phase transition boundary
    which is of the order of a constant; they are thus not required to
    grow to infinity together with the dimension (as assumed in most
    of the existing papers on high-dimensional factor models).
    Empirical Monte-Carlo study as well as the  analysis of  stock
    returns  data attest a very good performance of the proposed
    estimator.
    In all the tested cases, the new estimator largely outperforms
    the existing estimator using the same ratios of singular values.
  \end{abstract}



\end{frontmatter}

\section{Introduction}

\bll{Factor models have met a large success in data analysis in many
scientific fields such as psychology, economics and finance, 
signal processing, to name a few. One of the 
strengths of these  models relies on its capability to reduce 
the generally high dimension of the data to a much lower-dimensional 
common component. The structure of these models is complex
and many different versions of the models have been introduced so far
in the long-standing literature on the subject, ranging from 
{\em static} to {\em dynamic} or {\em generalized dynamic} factor
models on one hand, and from {\em exact} to {\em approximate} factor
models on the other hand.  
A recent survey  of this literature   can be found in \citet{Watson11}.
Efforts are however still paid to the study of these models because
unfortunately their  inference is not easy, especially
when the cross-sectional dimension $p$ 
and the temporal dimension $T$  are both large. 

In such high-dimensional context, 
the determination of the number $k$  of common factors in  
in a factor model is a challenging problem.
Misspecification of this number can deeply affect the quality of the
fitted factor model.  
In this context, the seminal paper \citet{BaiNg02} provided for the first time a
consistent estimator of $k$ for static factor models.   
This estimator has attracted much attention afterwards, and 
has  been improved or generalized, e.g.   
in \citet{BaiNg07} by the authors themselves, 
in \citet{Hallin07} for dynamic factor models and  in \citet{Alessi10}
for approximate factor models.  
It should be here mentioned  that as these developments  
mainly target at analysis of economic or financial data, 
the common factors in these  models are thought to be {\em
  pervasive}, or {\em strong},
in the sense that their strength is much higher than the strength
of the  idiosyncratic (error) component.  
The asymptotic consistency of the factor number estimator depends in a
large extent on this assumption.  However, some recent studies on
factor models suggest the importance for 
accommodating more factors in these models by including 
some {\em weaker factors}  which still have a
significant explanation power on both cross-sectional and temporal
correlations of the data. For example, \citet{Onatski15} makes a clear
distinction between strong factors and weak factors when considering
asymptotic approximations of the square loss function from a principal
components-based perspective.  A related work  allowing weak
factors can be found in \citet{Onatski12}.  

In this paper we consider a   factor model 
for high-dimensional time series proposed by \citet{LY012}:
the observations $Y$ is a $p\times T$ matrix with $p$ cross-sectional
units over $T$ time periods. Let $y_t$ denote the $p$-dimensional
vector observed at time $t$, then it consists of two components, 
a low dimensional  common-factor time series  $x_t$ and an
idiosyncratic component   $\varepsilon_t$:
\begin{equation}\label{model}
  y_t=Ax_t+\varepsilon_t,
\end{equation}
where $A$ is the factor loading matrix of size  $p\times k$ and 
$\{\varepsilon_t\}$ is a white noise sequence (temporal uncorrelated).
The factors in $(x_t)$ are here loaded contemporaneously; however this
is a time series and its temporal correlation implies that 
of the observations $\{y_t\}$. However, this
is the unique  source of temporal correlation,
and in this aspect, the model is  much more restrictive 
than the general dynamic models as introduced in 
\citet{Geweke77}, \citet{Sargen77}
and \citet{Forni00,Forni04,Forni05}.
Nevertheless, there are two advantages in this simplified
model. First, since potentially 
$(x_t)$ can be any kind of stationary time series of low dimension, the model can
already  cover a wide range of applications. 
Second, inference procedures are here more consistently defined 
and  more precise 
results can be expected, e.g. for the determination of the number of
factors. 
The factor model \eqref{model} can be considered as a good balance
between the generality of model coverage and the  technical
feasibility of underlying inference procedures.
}

The goal of this paper is to develop a powerful estimator of the
number of factors in the model \eqref{model}.
 \citet{LY012} proposed a ratio-based estimator
defined as follows.
Let  $\Sigma_y{ }=\cov(y_{t},y_{t-1})$ and $\Sigma_x{ }=\cov(x_{t},
x_{t-1})$ be  the lag-1 auto-covariance matrices of $y_t$ and $x_t$, respectively. Assuming
that
the factor and the noise are independent, we then have
\[
\Sigma_y{ }=A \Sigma_x{ }  A'~,
\]
which leads to its symmetric  counterpart
\begin{align}\label{y}
  M=\Sigma_y{ }\Sigma_y{ }' =A\Sigma_x{ }\Sigma_x{ }'A'~.
\end{align}
Since in general the $k\times k$ matrix $\Sigma_x{ }$ is of full rank
$k$,
the symmetric $p\times p$ matrix $M$ has exactly $k$
nonzero eigenvalues. Moreover,
the factor loading space $\mathcal{M}(A)$, i.e. the $k$-dimensional
subspace in $\R^p$ generated by the columns of $A$,
is spanned by the eigenvectors of $M$ corresponding to its nonzero
eigenvalues   $a_1\ge \cdots \ge a_k>0$ (factor eigenvalues). Let
\begin{equation}
  \label{Mhat}
  \widehat{M}=\widehat{\Sigma}_y{ }\widehat{\Sigma}_y{ }',
  \quad \text{where} ~
  \widehat{\Sigma}_y{ }=\frac1T\sum_{t=2}^{T+1}y_t y_{t-1}'~,
\end{equation}
be the sample counterparts of $M$ and $\Sigma_y{ }$, respectively.
The main observation is that
the $p-k$ null eigenvalues of $M$  will lead to $p-k$ ``relatively small''
sample eigenvalues in $\widehat M$, while the $k$ factor  eigenvalues
$(a_i)$
will generate $k$ ``relatively large'' eigenvalues in $\widehat M$.
This can be made very precisely in a classical {\em low-dimensional}
framework where we fix the dimension $p$ and let $T$ grow to
infinity: indeed by law of large numbers, $\widehat M\to M$ and by
continuity, all the eigenvalues $l_1\ge l_2\ge \cdots \ge l_p$
(sorted in decreasing order) of $\widehat M$ will converge to the
corresponding eigenvalues of $M$. In particular, for $k <i\le p$,
$l_i\to 0$ while $l_i \to a_i >0$ for $1\le i \le  k$.
Consider the ratio estimator
(\citet{LY012}):
\begin{equation}
  \label{eq:LY}
  \tilde{k}=\arg \min_{1\leq i < p} {l_{i+1}}/{l_i}.
\end{equation}
As $l_{k+1}/l_k$ will be the first ratio in this list which tends to
zero, $\tilde k$ will be a consistent estimator of $k$.

In the high-dimensional context however, $\widehat M$ will significantly deviate from
$M$  and the spectrum $(l_i)$ of $\widehat M$ will not be
close to that of $M$ anymore. In particular, the time for the
first minimum of the ratios in
\eqref{eq:LY} becomes noisy and can be much different from the target
value $k$.
Notice that the $k$ non-null factor eigenvalues $(a_i)$
are directly linked
to the strength of the factor time series $(x_t)$.
The precise
relationship between the ratios of sample eigenvalues
in \eqref{eq:LY} will ultimately depend on a complex interplay
between the strength of the factor eigenvalues  $(a_i)$ (compared to
the noise level), the dimension $p$ and the sample size $T$.

Despite of the introduction of
a very appealing ratio estimator \eqref{eq:LY},  
precise description of the sample ratios
${l_{i+1}}/{l_i}$ is missing in \cite{LY012}. Indeed, the authors establish the consistency
of the ratio estimator $\tilde k$ by requiring that the factor
strengths $(a_i)$ all explode  {\em at a same rate}:
$a_i\propto p^{\delta}$ for all $1\le i\le k$ and some $\delta>0$ as
the dimension $p$ grows to infinity. In other words, the factors
are all strong and they have a same asymptotic strength.
This limitation is quite severe because factors with different
levels  of strength cannot be all detected within this framework.
For instance, if we have factors with say three levels of strength
$p^{\delta_j}$, $j=1,2,3$ where $\delta_1>\delta_2>\delta_3 $, the
ratio estimator $\tilde k$ above will correctly  identify the group
of strongest factors $a_i\propto p^{\delta_1}$ while all the others
will be omitted.
In an attempt to correct such undesirable behavior, a two-step
estimation procedure is also proposed in \cite{LY012} which will
identify successively two groups of factors with top two strengths:
this means that in the  example above,
factors of strength  $a_i$' proportional to $p^{\delta_j}$ with
$j=1,2$ will be identified
while  the others will remain omitted.
The issue here is that {\it a priori}, we do not know how
many different levels of strength the factors could have and it is
unlikely we could  attempt to estimate such different levels as this
would lead to a  problem  that is equally  (if not more)
difficult than  the initial problem of
estimating the number of factors.

Inspired by the appropriateness of the ratio estimator $\tilde k$
\bl{in the high-dimensional context,
  the main objective of this paper is to provide a rigorous  theory
  for the estimation of the number of factors based on the ratios
  $\{l_{i+1}/l_i\}$
  under the high-dimensional setting  where $p $ and $T$ tend to
  infinity proportionally.

  The paper contains  two main contributions.
  First, {\em we characterize completely the limits of both the factor
    eigenvalues $\{l_i\}_{1\le i\le k}$ and the noise eigenvalues
    $\{l_i\}_{k <i\le p}$.}
  For the noise part, as $k$ (although unknown)
  is much smaller than   the dimension $p$,
  we prove that the spectral distribution generated by $\{l_i\}_{k <i\le p}$
  has a limit which coincide with the limit of the spectral
  distribution generated by the $p$ eigenvalues of the (unobserved) matrix
  $\widehat M_\veps= \hat\Sigma_\veps\hat\Sigma'_\veps$ where
  $\Sigma_\veps=T^{-1}\sum_{t=2}^{T+1}\veps_t\veps'_{t-1}$.
  This limiting distribution has been explored elsewhere in
  \citet{li} and its support found to be
  a  compact  interval $[a,b]$.
  As for the factor part $\{l_i\}_{1\le i\le k}$, although it is
  highly expected that they should have a limit located outside the
  base interval $[a,b]$, we establish a
  {\em phase transition phenomenon}: a factor eigenvalue $l_i$ will
  tend to a limit $\lambda_i >b $ (outlier) {\em if and only if} the
  corresponding population factor strength $a_i$
  exceeds some critical
  value $\tau$. In other words, if a factor $a_i$  is too weak, then the
  corresponding sample factor eigenvalue $\lambda_i$ will tend to $b$, the
  (limit of) maximum of the noise eigenvalues and it will be hardly
  detectable.
  Moreover, both the outlier limits $\{\lambda_i\}$ and the critical value
  $\tau$ are characterized through the model parameters.

  The second main contribution of the paper is {\em the derivation of a new
  estimator $\hat k$ of the number of factors}  based on the finding
  above.
  If $k_0$  denotes the number of {\em significant factors}, i.e. with
  factor strength $a_i>\tau$, then using an appropriate thresholding
  interval  $(1-d_T,1)$ for the sample ratios $\{l_{i+1}/l_i\}$, the
  derived estimator $\hat k$  is strongly consistent converging to $k_0$.
  In addition to this well-justified consistency, the main advantage
  of the proposed estimator is its robustness against
  possibly multiple levels of factor strength;
  in theory, all factors with strength above the constant $\tau$ are
  detectable. 
  \bll{Therefore,  both strong factors and weak
    factors can be present, and 
    their strengths  can 
    have different
    asymptotic rates with regard to the dimension  $p$ in order to be
    detected from the observed samples.
    This is a key difference between the method provided in this
    paper  and most of existing estimators of the factor number 
    as recalled previously (the reader is however reminded that the
    model \eqref{model} is more restrictive than a general dynamic
    factor model).
    Notice however that these 
    precise results  have been obtained at the cost of  some  drastic
    simplification of the idiosyncratic component 
    $\{\veps_t\}$, namely  
    independence has been assumed both serially and cross-sectionally (over the
    time and the dimension), and the components are normalized to have a
    same value of  variance (see Assumption 2 in Section~\ref{ModAss}).
    These limitations are required by the technical tools employed in this
    paper and  some non-trivial extension of
    these  tools 
    are needed to get rid of these limitations.
  }

  From a methodological point of view,
  our approach is based on recent advances from random matrix theory,
  specifically on the so-called spiked population models or
  more generally on finite-rank perturbations of large  random matrices.
  We start by  identifying
  the sample matrix $\widehat M$ as a finite-rank
  perturbation of the  base matrix $\widehat M_\veps$ associated to
  the noise.
  In a recent paper \citet{li}, the limiting spectral distribution  of
  the eigenvalues
  of $\widehat M_\veps$ has been found and the base interval $[a,b]$
  characterized.
  By developing the mentioned perturbation theory for the
  autocovariance matrix $\hat M$, we find the characterization of the
  limits of its eigenvalues $\{l_i\}$.

  For the strong consistency of the proposed ratio estimator $\hat k$,
  a main ingredient is
  the almost sure convergence of the largest eigenvalue
  of the base matrix $\widehat M_\veps$ to the right edge $b$,  recently
  established in   \citet{WangYao14}.  This result serves as the
  cornerstone for distinguishing between significant factors and noise
  components.
}

It is worth mentioning a related paper
\citet{Onatski10} where the author stands from a similar
perspective with the  method in this paper. However,
that paper addresses  static
approximate factor models without time series dependence
and more importantly,
the assumption of explosion  of all 
factor eigenvalues is still required
 which,
on the contrary, is released in this paper. 
\bll{Other related references
include \citet{Jin14} and \citet{Wang15} 
where the limiting spectral distribution and the strong convergence of
extreme eigenvalues are derived for the matrix 
$\frac12({\widehat\Sigma}_\varepsilon+{\widehat\Sigma}'_\varepsilon)$. One should mention
that these works are more related to the study 
in \citet{li} and \citet{WangYao14} on spectral limits of the matrix 
${\widehat M}_\varepsilon$, and they have  no results either on convergence of the spiked
(factor) eigenvalues  or on the estimation of the number
of factors as proposed in this paper.
}


\bl{
  The rest of the paper is organized as follows.
  In Section~\ref{ModAss}, after introduction of the model assumptions
  we develop our first main result regarding spectral limits
  of  $\hat M$.
  The new estimator $\hat k$ is then introduced in
  Section~\ref{estimation} and its strong  convergence to the number
  of significant factors $k_0$ established.
  In
  Section \ref{simulation}, detailed Monte-Carlo experiments are
  conducted to check the finite-sample properties of the proposed
  estimator and to compare it with the ratio estimator $\tilde k$
  \eqref{eq:LY} from \citet{LY012}.
  Both estimators $\tilde k$ and $\hat k$ are
  next tested in Section \ref{application} on a real
  data set from Standard \& Poor stock returns
  and compared in details.
  Notice that  
  some technical lemmas used in the main
  proofs are gathered in a companion paper of supplementary material \citet{li2}.
}


\section{Large-dimensional limits of noise and factor eigenvalues}
\label{ModAss}

The static factor
model \eqref{model} is further specified to satisfy the following
assumptions.

\medskip\noindent{\bf Assumption 1.}\quad
The factor $(x_t)$ is a $k$-dimensional ($k\ll p$ fixed) stationary time series, each dimension is independent of each other, with the representation of each component:
\[
x_{it}=\sum_{l=0}^{\infty}\phi_{il}\eta_{i\,t-l}~,i=1,\cdots,k,~t=1,\cdots,T,
\]
where $(\eta_{ik})$ is a real-valued and weakly stationary white noise
with mean 0 and variance $\sigma_i^2$. The series $\{x_{it}\}$ has
variance $\gamma_0(i)$ and lag-1 auto-covariance $\gamma_1(i)$.
\bl{
  Moreover, the variance $\gamma_0(i)$ will be hereafter referred as the
  {\em strength} of the $i$-th factor time series $\{x_{it}\}$.}

\medskip\noindent{\bf \re{Assumption 2.}}\quad
The idiosyncratic component $(\varepsilon_t)$ is independent of
$x_t$. $\varepsilon_t$ is $p-$dimensional real valued random vector
with independent entries $\varepsilon_{it},~i=1,\cdots,p$,
\bl{not necessarily identically distributed}, satisfying
\[
\E (\varepsilon_{i\,t})=0,~\E (\varepsilon_{i\,t}^2)=\sigma^2,
\]
and for any $\eta>0$,
\begin{equation} \label{lindeberg}
\frac{1}{\eta^4 pT}\sum_{i=1}^p \sum_{t=1}^{T+1}\E
  \left(|\varepsilon_{i\,t}|^4I_{(|\varepsilon_{i\,t}|\geq \eta
      T^{1/4})}\right) \longrightarrow 0  \quad \text{as } (pT) \to\infty.
\end{equation}

\medskip\noindent{\bf Assumption 3.}\quad
The dimension $p$ and the sample size $T$  tend to infinity
proportionally:
$p\to\infty$, $T=T(p) \to\infty$ and
$p/T \rightarrow y>0$.

\medskip

Assumption 1 defines the static  factor model considered in
this paper. Assumption 2 details the moment condition and the
independent structure of the noise. In particular, \eqref{lindeberg}
is a Lindeberg-type condition widely used in random matrix theory.
\bl{In particular, if the fourth moments of the variables
  $\{\veps_{it}\}$ are uniformly bounded, the Lindeberg condition is
  satisfied}.
Assumption 3 defines the high-dimensional setting where the dimension
and the sample size can be both large without however one dominating
the other.


First we have,
\begin{align*}
  \hat{\Sigma}_y{ } &= \frac{1}{T}\sum_{t=2}^{T+1}y_ty_{t-1}'
  = \frac{1}{T}\sum_{t=2}^{T+1}(Ax_t+\veps_t) (Ax_{t-1}+\veps_{t-1})'\\
  &= \frac{1}{T}\sum_{t=2}^{T+1}Ax_tx_{t-1}'A' +
  \frac{1}{T} \sum_{t=2}^{T+1} \lb   Ax_t\veps_{t-1}'+  \veps_t x_{t-1}' A' \rb
  + \frac{1}{T}\sum_{t=2}^{T+1}\veps_t\veps_{t-1}'\\
  &:= {P}_A+\widehat{\Sigma}_\varepsilon{ }.
\end{align*}
The matrix $\widehat{\Sigma}_\varepsilon{ } = T^{-1} \sum_t \veps_t\veps_{t-1}'$ is the analogous sample
autocovariance matrix associated to the noise $(\veps_t)$.
Since $A$ has rank $k$, the rank of the matrix $P_A$
is bounded by $2k$ (we will see in fact that asymptotically, the rank
of $P_A$ will be eventually $k$).
Therefore, the  autocovariance matrix of interest
\re{$\widehat{\Sigma}_y{}$}
is seen as a
finite-rank perturbation of  the noise autocovariance matrix
$\widehat{\Sigma}_{\varepsilon}{ }$. \re{Since the matrix $\widehat{\Sigma}_y{ }$ is not symmetric, we consider its singular values, which are also the square root of the eigenvalues of $\widehat{M}:=\widehat{\Sigma}_y{ }\widehat{\Sigma}^{'}_y{ }$. Therefore, the study of the singular values  of $\widehat{\Sigma}_y{ }$ reduces to the study of the eigenvalues of $\widehat{M}$, which is also a finite rank perturbation of the base component $\widehat{M}_{\varepsilon}:=\widehat{\Sigma}_{\varepsilon}{ }\widehat{\Sigma}^{'}_{\varepsilon}$.}

Finite-rank perturbations of random matrices have been actively
studied in recent years and the theory is much linked to the {\em
  spiked population models} well known in high-dimensional statistics
literature. For some recent accounts on this theory,  we refer to
\citet{Johnstone01},
\citet{BaikSilv06},
\citet{BaiYao08},
\citet{BenaychNadakuditi11},
\citet{DY12}
and the references therein.
A general picture from this theory is that
first,  the \re{eigenvalues}  of the base
matrix will converge to a limiting spectral distribution (LSD)
with a compact support, say an interval $[a,b]$; and secondly, for the \re{eigenvalues} of the
perturbed matrices,
most of them (base eigenvalues)
will converge to the same LSD independently of the
perturbation while a small number among the largest
ones  will converge to a limit outside the support of the LSD
(outliers).
However, all the existing literature cited above concern \re{the finite rank perturbation of large-dimensional } sample
covariance matrices or Wigner matrices.
As a theoretic contribution of  the paper, we extend this theory to
the case of a perturbed auto-covariance matrix by giving exact
conditions under which the aforementioned dichotomy between
base eigenvalues and outliers still hold.
Specifically, it will be proved in this section that once the $k$
factor strengths $(a_i)$ are not ``too weak'',
they will generate exactly $k$ outliers, while
the remaining $p-k$ \re{eigenvalues}  will behave as the \re{eigenvalues}
of the base \re{$\widehat{M}_{\varepsilon}$, which converges} to a compactly supported LSD.

It is then apparent that under such dichotomy and by ``counting'' the
outliers  outside the interval $[a,b]$, we will be able to obtain  a
consistent estimator of the number of factors $k$.

In what follows, we first recall two  existing result
on the asymptotic of the singular values of
$\widehat{\Sigma}_\varepsilon{ }$.
Then we develop our theory on the limits of largest
(outliers)   and base singular values of $\widehat{\Sigma}_y{ }$.

\subsection{Limiting spectral distribution of
  \re{$\widehat{M}_\varepsilon{ }$}}

We first recall two useful results on the base matrix
$\widehat{M}_\varepsilon{ }$.
Firstly,
the limiting spectral distribution of the \re{eigenvalues of
$\widehat{M}_\varepsilon{ }$}
has been obtained in a recent paper \citet{li}.
Write
\begin{align*}
  \widehat{M}_{\varepsilon}
  & =\lb\frac 1T\sum_{t=2}^{T+1}\varepsilon_t
  \varepsilon_{t-1}'\rb\lb\frac 1T\sum_{t=2}^{T+1}\varepsilon_t
  \varepsilon_{t-1}'\rb'\\
  &=\frac{1}{T^2}XY'YX'~,
\end{align*}
with  the data matrices
\[X=\left(
  \begin{array}{ccc}
    \veps_{12} &\cdots &\veps_{1,T+1}\\
    \vdots&\ddots&\vdots\\
    \veps_{p2} &\cdots &\veps_{p,T+1}
  \end{array}\right), \quad
Y=\left(
  \begin{array}{ccc}
    \veps_{11} &\cdots &\veps_{1T}\\
    \vdots&\ddots&\vdots\\
    \veps_{p1} &\cdots &\veps_{pT}
  \end{array}
\right).\]

Furthermore, let
$\mu$ be a measure on  the real line supported on an interval
$[\alpha,\beta]$ (the end points can be infinity), \re{with its}
  Stieltjes transform  defined as
\[
m(z)=\int \frac{1}{t-z} d\mu(t)~\quad \mbox{for} ~z \in
\mathbb{C}\backslash \text{supp}(\mu)~,
\]
and
its $T$-transform  as
\[
T(z)=\int \frac{t}{z-t}d \mu(t)~\quad \mbox{for} ~z\in
\mathbb{C}\backslash \text{supp}(\mu).
\]
\re{Notice here that the T-transform }
 is a decreasing homeomorphism  from $(-\infty,\alpha)$ onto
$(T(\alpha^{-}),0)$ and from $(\beta,+\infty)$ onto $(0,T(\beta^{+}))$, which related to each other by
the following \re{equation}:
\[
T(z)=-1-zm(z)~.
\]

\begin{proposition}\label{LSD}[\citet{li}]
  \smallskip
  \par Suppose that Assumptions 2 and 3 hold \re{with} $\sigma^2=1$. Then,
  the \re{empirical spectral  distribution of}
  $B:=1/T^2Y'YX'X$ (\re{ which is the  companion matrix  of  $\widehat M_\veps$})
  converges a.s. to a non-random limit F, whose Stieltjes
  transform $m=m(z)$  satisfies the equation
  \begin{align}\label{str}
    z^2m^3(z)-2z(y-1)m^2(z)+(y-1)^2m(z)-zm(z)-1=0~.
  \end{align}
  In particular, this LSD is supported on the interval
  $[a\mathbf{1}_{\{y \geq 1\}},b]$
  whose  end points are
  \begin{align}
    &a=\left(-1+20y+8y^2-(1+8y)^{3/2}\right)/8~,\\
    &b=\left(-1+20y+8y^2+(1+8y)^{3/2}\right)/8~.\label{b}
  \end{align}
\end{proposition}

Notice that the companion matrix $B$  is $T\times T$ and
it shares the same $p\wedge T$ non-null eigenvalues as
$\widehat M_\veps$, \re{therefore, the support of $\widehat{M}_{\varepsilon}$ is also $[a,b]$} The LSD $F$ of $B$ and the LSD $F^*$  of $\widehat
M_\veps$  are linked by the relationship
\[   y F^* -  F = (y-1) \delta_0~,
\]
where $\delta_0$ is the Dirac mass at the origin.

\begin{remark}
  The equation \eqref{str} can be expressed using the $T$-transform:
  \begin{align}\label{ttr}
    \left(T(z)+1\right)\left(T(z)+y\right)^2=zT(z)
  \end{align}
\end{remark}

The second result is about the convergence of the largest eigenvalue of
$\widehat M_\veps$.

\begin{proposition}\label{lmax}[\citet{WangYao14}]
  \smallskip
  \par Suppose that Assumptions 2 and 3 hold \re{with $\sigma^2=1$}. Then,
  the largest eigenvalue of
  \re{$\widehat{M}_{\varepsilon}$}
  converges a.s. to the right end point $b$ of its LSD given in
  \eqref{b}.
\end{proposition}

Combining Propositions~\ref{LSD} and \ref{lmax}, we have \re{the following corollary.}
\begin{corollary}\label{coro}
 Under the same conditions as in Proposition~\ref{lmax}, if
 $(\beta_j)$  are sorted eigenvalues of $\widehat M_\veps$, then
 for any fixed $m$, the $m$  largest eigenvalues
 $\beta_1\ge  \beta_2\ge \cdots\ge  \beta_m$ all converge to $b$.
\end{corollary}
\begin{proof}
  For any $\delta>0$, almost surely the number of sample eigenvalues
  of $\beta_j$ falling  into
  the interval $(b-\delta, b)$ grows to infinity
  due to the fact the density of the LSD is positive and continuous
  on this interval. Then for fixed
  $m$, a.s. $\liminf_{p\to\infty}   \beta_m\ge b-\delta$. By letting
  $\delta\to 0$, we have a.s.  $\liminf_{p\to\infty}   \beta_m\ge b$.
  Obviously, $\limsup_{p\to\infty}   \beta_m \le \limsup_{p\to\infty}
  \beta_1= b$, i.e a.s. $\lim_{p\to\infty}   \beta_m=b$.
\end{proof}

\subsection{\re{Convergence} of the largest
  \re{eigenvalues} of the
  sample autocovariance matrix
  $\widehat{M}$}

The following main result of the section characterizes the limits of
the $k$-largest \re{eigenvalues} of the sample autocovariance matrix
\re{$\widehat{M}$}.

\begin{theorem}\label{mainth}
  Suppose that the model \eqref{model} satisfies
  Assumptions 1, 2 and 3 and that 
  the noise $\{\varepsilon_t\}$  are  normal distributed
  and \bll{the loading matrix $A$ is normalized as $A'A=I_k$.}
  Let $l_i~(1\leq i\leq k)$ denote the $k$
  largest eigenvalue of
  $\widehat{M}$. Then for
  each $1\le i\le k$,
  $l_i/\sigma^4$ converges almost surely to a limit $\lambda_i$. Moreover,
  \[
  \lambda_i=b  \quad \text{when} ~~   T_1(i) \ge T(b^+),
  \]
  where
  \begin{equation}\label{t1}
    T_1(i)=\frac{2y\sigma^2\gamma_0(i)+\gamma_1(i)^2-\sqrt{(2y\sigma^2\gamma_0(i)+\gamma_1(i)^2)^2-4 y^2\sigma^4(\gamma_0(i)^2-\gamma_1(i)^2)}}{2\gamma_0(i)^2-2\gamma_1(i)^2}~.
  \end{equation}
  Otherwise, i.e.
  $T_1(i)< T(b^+)$,  $\lambda_i>b$ and its value is characterized by
  the fact that the
  $T$-transform  $T(\lambda_i)$
  is the solution to  the equation:
  \begin{equation}
    \label{limit}
    \left(y\sigma^2-\gamma_0(i)T(\lambda_i)\right)^2=\gamma_1(i)^2T(\lambda_i)\left(1+T(\lambda_i)\right)~.
  \end{equation}
\end{theorem}

%


\bl{
  The theorem establishes a {\em phase transition phenomenon}
  for the $k$ sample factor  eigenvalues $(l_i)$.
  Define the {\em  number of  significant factors}
  \begin{equation}
    \label{k0}
    k_0 = \sharp\left\{ 1\le i\le k : ~
      T_1(i)< T(b^+)\right\}.
  \end{equation}
  Therefore, for each of the $k_0$ significant factor,
  the corresponding sample eigenvalue $l_i$ will converge to a limit
  $\lambda_i$ outside the base support interval $[a,b]$.
  In contrary, for the $k-k_0$ factors for which
  $T_1(i)\ge T(b^+)$, they are too weak in the sense that
  the corresponding sample eigenvalue $l_i$ will converge
  to $b$ which is also the limit of the largest
  noise eigenvalues $l_{k+1},\ldots, l_{k+m}$ ($m$ is a fixed
  number here). Therefore, these weakest factors will be merged 
  with noise component and their detection becomes hardly possible.

  Later in Section~\ref{condition}, it will be established
  that for the $i$-th factor time series be significant,
  the phase  transition condition $T_1(i)< T(b^+)$ essentially requires the strength
  $\gamma_0(i)$ be large enough.
}

\begin{proof}{(of Theorem~\ref{mainth})}\quad
  The proof consists in four steps where 
  some technical lemmas are
  to be found in  the companion paper of 
  supplementary material \citet{li2}.

\medskip\noindent{\textsc{Step 1.  Simplification of variance $\sigma^2$ of white noise $\{\varepsilon_{i\,t}\}$.}}\quad  To ease the complexity of the proof of this main theorem, we firstly reduce the variance of the white noise from $\sigma^2$ to 1. In our model setting, we have \eqref{model} equivalent to
\[
\frac{y_t}{\sigma}=A\frac{x_t}{\sigma}+\frac{\varepsilon_t}{\sigma}~.
\]
And if we denote $\tilde{y}_t=y_t/\sigma$, $\tilde{x}_t=x_t/\sigma$ and $\tilde
{\varepsilon}_t=\varepsilon_t/\sigma$, then we are dealing with the model
\begin{equation}\label{model1}
\tilde{y}_t=A\tilde{x}_t+\tilde{\varepsilon}_t~,
\end{equation}
where the white noise $\tilde{\varepsilon}$ has mean zero and unit variance and the variance and autocovariance of the factor process $\{\tilde{x}_t\}$ satisfies
\begin{align}
\tilde{\gamma}_0(i)=\gamma_0(i)/\sigma^2~,
\tilde{\gamma}_1(i)=\gamma_1(i)/\sigma^2~,
\end{align}
\re{in which $\gamma_0(i)$ and $\gamma_1(i)$ are the variance and autocovariance of the original factor process $\{x_t\}$.}
\noindent
Therefore, in all the following, we just consider the standardized Model \eqref{model1}. For convenience, we use notations of the original model \eqref{model} and set $\sigma^2=1$ to investigate Model \eqref{model1}. At the end of the proof, we will replace the value of $\gamma_0(i)$ and $\gamma_1(i)$ with $\tilde{\gamma}_0(i)$ and $\tilde{\gamma}_1(i)$ to recover the corresponding results for Model \eqref{model}.

  \medskip\noindent{\textsc{Step 2.  Simplification of matrix $A$.}}\quad
  Here we argue that it is enough to consider the case where the
  loading matrix $A$ has the canonical form
  \[
  A=\begin{pmatrix}
    I_k \\
    \bm{0}_{p-k} \\
  \end{pmatrix}~.
  \]

  \bll{
  Indeed, suppose $A$ is not in this canonical form. Since by
  assumption $A'A=I_k$, we can complete $A$ to an orthogonal matrix 
  $Q=(A,C)$ by adding appropriate orthonormal columns. From the model 
  equation \eqref{model}, we have 
  \begin{align*}
    Q'y_t=Q'Ax_t+Q'\varepsilon_t=\begin{pmatrix}
      A' \\
      C' \\
    \end{pmatrix}Ax_t+Q'\varepsilon_t=\begin{pmatrix}
      I_k \\
      \bm{0}_{p-k} \\
    \end{pmatrix}x_t+Q' \varepsilon_t~.
  \end{align*}
  Since $\varepsilon_t \sim \N(0,I_p)$ and $Q'$ is orthogonal,
  $Q'\varepsilon_t \sim \N(0,I_p)$.
  Let $z_t:=Q'y_t$, then  $z_t$  satisfies the model equation 
  \eqref{model} with a canonical loading matrix.
  What happens is that    the singular values of 
  the two lag-1 autocovariance matrices 
  \[
  \frac 1T \sum_{t=2}^{T+1} z_tz_{t-1}', 
  \quad 
  \frac 1T \sum_{t=2}^{T+1} y_ty_{t-1}'
  \]
  are the same: this is 
  simply due to  fact that
}
  \begin{align*}
    &~~~~\left(\frac 1T \sum_{t=2}^{T+1} y_ty_{t-1}'\right)\left(\frac 1T \sum_{t=2}^{T+1} y_ty_{t-1}'\right)'\\
    &=\left(\frac 1T \sum_{t=2}^{T+1} Q'z_t\cdot(Q'z_{t-1})'\right)\left(\frac 1T \sum_{t=2}^{T+1} Q'z_t\cdot(Q'z_{t-1})'\right)'\\
    &=Q'\left(\frac 1T \sum_{t=2}^{T+1} z_tz_{t-1}'\right)\left(\frac 1T \sum_{t=2}^{T+1} z_tz_{t-1}'\right)'Q~.\\
  \end{align*}

  \medskip\noindent\textsc{Step 3. Derivation of the main equation \eqref{limit}}\quad
  From now on we assume that $A$ is in its canonical form.
  By  the definition of $y_t$, we have
  \begin{align}
    &~~~~\widehat{\Sigma}_y{ }=\frac 1T \sum_{t=2}^{T+1} y_t y_{t-1}'\nonumber \\
    &=\frac 1T \sum_{t=2}^{T+1}\begin{pmatrix}
      x_{1\,t}+\varepsilon_{1\,t} \\
      \vdots\\
      x_{k\,t}+\varepsilon_{k\,t} \\
      \varepsilon_{k+1\,t} \\
      \vdots \\
      \varepsilon_{p\,t} \\
    \end{pmatrix}\begin{pmatrix}
      x_{1\,t-1}+\varepsilon_{1\,t-1} &\cdots & x_{k\,t-1}+\varepsilon_{k\,t-1}& \varepsilon_{k+1\,t-1} & \cdots & \varepsilon_{p\,t-1} \\
    \end{pmatrix}\nonumber\\
    &:= \begin{pmatrix}
      A & B \\
      C & D \\
    \end{pmatrix}~,\label{matrixdefinition}
  \end{align}
  where we use $A$, $B$, $C$ and $D$ to denote the four blocks. Besides, if we use the notation:
  \begin{align*}
    &X_0:=\frac{1}{\sqrt T}\begin{pmatrix}
      x_{1\,1}+\varepsilon_{1\,1} & \cdots & x_{k\,1}+\varepsilon_{k\,1} \\
      \vdots & \ddots & \vdots\\
      x_{1\,T}+\varepsilon_{1\,T} & \cdots & x_{k\,T}+\varepsilon_{k\,T} \\
    \end{pmatrix}~,\\
    &X_1:=\frac{1}{\sqrt T}\begin{pmatrix}
      x_{1\,2}+\varepsilon_{1\,2} & \cdots & x_{k\,2}+\varepsilon_{k\,2} \\
      \vdots & \ddots & \vdots\\
      x_{1\,T+1}+\varepsilon_{1\,T+1} & \cdots & x_{k\,T+1}+\varepsilon_{k\,T+1} \\
    \end{pmatrix}~,\\
    &E_1:=\frac{1}{\sqrt T}\begin{pmatrix}
      \varepsilon_{k+1\,1} & \cdots & \varepsilon_{p\,1} \\
      \vdots & \ddots & \vdots \\
      \varepsilon_{k+1\,T} & \cdots & \varepsilon_{p\,T} \\
    \end{pmatrix},~~\quad
    E_2:=\frac{1}{\sqrt T}\begin{pmatrix}
      \varepsilon_{k+1\,2} & \cdots & \varepsilon_{p\,2} \\
      \vdots & \ddots & \vdots \\
      \varepsilon_{k+1~T+1} & \cdots & \varepsilon_{p~T+1} \\
    \end{pmatrix}~,
  \end{align*}
  then we have
  \begin{eqnarray}\label{bd}
    A=X_1'X_0~,~~B=X_1'E_1~,~~C=E_2'X_0~,~~D=E_2'E_1~.
  \end{eqnarray}

  Since $l$ is the extreme large eigenvalue of $\widehat{\Sigma}_y{ }\widehat{\Sigma}_y{ }'$, $\sqrt{l}$ is the extreme large singular value of $\widehat{\Sigma}_y{ }$, which is also equivalent to saying that $\sqrt{l}$ is  the positive eigenvalue of the $2p \times 2p$ matrix
  \begin{equation}\label{yy}
    \left(
      \begin{array}{cc}
        0 & \widehat{\Sigma}_y{ }\\
        \widehat{\Sigma}_y{ }' & 0\\
      \end{array}\right) ~.
  \end{equation}
  And use the block expression \eqref{matrixdefinition}, combining with the definition of each block in \eqref{bd}, \eqref{yy} is equivalent to
  \begin{equation}\label{matrix1}
    \left(
      \begin{array}{cccc}
        0 & 0 & X_1'X_0 & X_1'E_1 \\
        0 & 0 & E_2'X_0 & E_2'E_1 \\
        X_0'X_1 & X_0' E_2 & 0 & 0 \\
        E_1'X_1 & E_1'E_2 & 0 & 0 \\
      \end{array}
    \right)~.
  \end{equation}
  If we interchange the second and third row block and column block in \eqref{matrix1}, its eigenvalues remain the same. Therefore, $\sqrt{l}$ should satisfy the following equation
  \begin{equation}\label{matrix2}
    \left\lvert\left(
        \begin{array}{cccc}
          \sqrt{l} & -X_1'X_0 & 0 & -X_1'E_1 \\
          -X_0'X_1 & \sqrt{l} & -X_0'E_2 & 0 \\
          0 & -E_2'X_0 & \sqrt{l} & -E_2'E_1 \\
          -E_1'X_1 & 0 & -E_1'E_2 & \sqrt{l} \\
        \end{array}\right)\right\rvert=0~.
  \end{equation}
  Then for block matrix, we have the identity $\det\begin{pmatrix}
    A & B \\
    C & D \\
  \end{pmatrix}=\det D \cdot \det(A-BD^{-1}C)$ when $D$ is invertible, then \eqref{matrix2} is equivalent to

  \begin{align}\label{ll1}
    \left\rvert\left(\begin{smallmatrix}
          \sqrt{l} & -X_1'X_0 \\
          -X_0'X_1 & \sqrt{l} \\
        \end{smallmatrix}\right)-\left(\begin{smallmatrix}
          0 & -X_1'E_1 \\
          -X_0'E_2& 0 \\
        \end{smallmatrix}\right)\left(\begin{smallmatrix}
          \sqrt{l} & -E_2'E_1 \\
          -E_1'E_2 & \sqrt{l} \\
        \end{smallmatrix}\right)^{-1}\left(\begin{smallmatrix}
          0 & -E_2'X_0 \\
          -E_1'X_1 & 0 \\
        \end{smallmatrix}\right)
    \right\rvert=0~,\nonumber\\
    \setlength{\arraycolsep}{0.2pt}
  \end{align}
  which is due to the fact that $\sqrt{l}$ is the extreme singular value, then
  \[\left\rvert
    \begin{pmatrix}
      \sqrt{l} & -E_2'E_1 \\
      -E_1'E_2 & \sqrt{l} \\
    \end{pmatrix}\right\rvert\neq 0
  \]and therefore is invertible.

  Then if we do the calculation of \[
  \begin{pmatrix}
    \sqrt{l} & -E_2'E_1 \\
    -E_1'E_2 & \sqrt{l}\\
  \end{pmatrix}^{-1}~,
  \] \eqref{ll1} is equivalent to
  \[\left\rvert
    \left(\begin{smallmatrix}
        \sqrt{l} I_k-\sqrt{l} X_1'E_1(lI-E_1'E_2E_2'E_1)^{-1}E_1'X_1 & -X_1'\left(I+E_1E_1'E_2(lI-E_2'E_1E_1'E_2)^{-1}E_2'\right)X_0\\
        -X_0'\left(I+E_2E_2'E_1(lI-E_1'E_2E_2'E_1)^{-1}E_1'\right)X_1&
        \sqrt{l}I_k-\sqrt{l} X_0'E_2(lI-E_2'E_1E_1'E_2)^{-1}E_2'X_0 \\
      \end{smallmatrix}\right)\right\rvert=0~,
  \]
  and using the simple fact that
  \[
  A(l I-BA)^{-1}=(l I-AB)^{-1}A~
  \]
  leads to
  \begin{align}\label{m1}\left\rvert
      \left(\begin{smallmatrix}
          \sqrt{l} I_k-\sqrt{l} X_1'(lI-E_1E_1'E_2E_2')^{-1}E_1E_1'X_1 & -X_1'\left(I+(lI-E_1E_1'E_2E_2')^{-1}E_1E_1'E_2E_2'\right)X_0\\
          -X_0'\left(I+(lI-E_2E_2'E_1E_1')^{-1}E_2E_2'E_1E_1'\right)X_1&
          \sqrt{l}I_k-\sqrt{l} X_0'(lI-E_2E_2'E_1E_1')^{-1}E_2E_2'X_0 \\
        \end{smallmatrix}\right)\right\rvert=0~.\nonumber\\
  \end{align}
Taking Lemmas 1.3 and 1.4 
given in \citet{li2} into consideration, the matrix in \eqref{m1} tends to a block matrix:
  \[
  \left(
    \begin{array}{cc}
      \begin{matrix}
        \frac{\sqrt{\lambda} (y-\gamma_0(1)T(\lambda))}{y+T(\lambda)} & \cdots & 0  \\
        \vdots & \ddots & \vdots\\
        0& \cdots  & \frac{\sqrt{\lambda} (y-\gamma_0(k)T(\lambda))}{y+T(\lambda)}
      \end{matrix} & \begin{matrix}
        -(1+T(\lambda))\gamma_1(1) & \cdots & 0  \\
        \vdots & \ddots & \vdots\\
        0& \cdots  & -(1+T(\lambda))\gamma_1(k)
      \end{matrix}\\[8mm]
      \hdashline
      \\
      \begin{matrix}
        -(1+T(\lambda))\gamma_1(1) & \cdots & 0  \\
        \vdots & \ddots & \vdots\\
        0& \cdots  & -(1+T(\lambda))\gamma_1(k)
      \end{matrix}  &  \begin{matrix}
        \frac{\sqrt{\lambda} (y-\gamma_0(1)T(\lambda))}{y+T(\lambda)} & \cdots & 0  \\
        \vdots & \ddots & \vdots\\
        0& \cdots  & \frac{\sqrt{\lambda} (y-\gamma_0(k)T(\lambda))}{y+T(\lambda)}
      \end{matrix}
    \end{array}
  \right)~,
  \]
  so $\lambda$ should make the determinant of this matrix equal to $0$.
  If we interchange the first and second column block, the matrix becomes the following:
  \[
  \left(
    \begin{array}{c:c}
      \begin{matrix}
        -(1+T(\lambda))\gamma_1(1) & \cdots & 0  \\
        \vdots & \ddots & \vdots\\
        0& \cdots  & -(1+T(\lambda))\gamma_1(k)
      \end{matrix} & \begin{matrix}
        \frac{\sqrt{\lambda} (y-\gamma_0(1)T(\lambda))}{y+T(\lambda)} & \cdots & 0  \\
        \vdots & \ddots & \vdots\\
        0& \cdots  & \frac{\sqrt{\lambda} (y-\gamma_0(k)T(\lambda))}{y+T(\lambda)}
      \end{matrix}\\[8mm]
      \hdashline
      \\
      \begin{matrix}
        \frac{\sqrt{\lambda} (y-\gamma_0(1)T(\lambda))}{y+T(\lambda)} & \cdots & 0  \\
        \vdots & \ddots & \vdots\\
        0& \cdots  & \frac{\sqrt{\lambda} (y-\gamma_0(k)T(\lambda))}{y+T(\lambda)}
      \end{matrix} & \begin{matrix}
        -(1+T(\lambda))\gamma_1(1) & \cdots & 0  \\
        \vdots & \ddots & \vdots\\
        0& \cdots  & -(1+T(\lambda))\gamma_1(k)
      \end{matrix}
    \end{array}
  \right)~.
  \]
  Since the diagonal block
  \[
  \left\vert\begin{pmatrix}
      -(1+T(\lambda))\gamma_1(1) & \cdots & 0  \\
      \vdots & \ddots & \vdots\\
      0& \cdots  & -(1+T(\lambda))\gamma_1(k)
    \end{pmatrix}\right \vert \neq 0~,
  \]
  we can use  the identity
  \[
  \det\begin{pmatrix}
    A & B \\
    C & D \\
  \end{pmatrix}=\det D \cdot \det(A-BD^{-1}C)\]
  again,
  and this leads to the result:
  \[
  \lambda\left(y-\gamma_0(i)T(\lambda)\right)^2-\gamma_1(i)^2\left(1+T(\lambda)\right)^2\left(y+T(\lambda)\right)^2=0~,~~~~i \in [1, \cdots, k]~.
  \]
  Combining this equation with \eqref{ttr} and replacing $\gamma_0(i), \gamma_1(i)$ with $\gamma_0(i)/\sigma^2, \gamma_1(i)/\sigma^2$ leads to the equation~\eqref{limit}.

  \medskip\noindent{\textsc{Step 4. Derivation of the condition $T_1(i)
      < T(b^+)$.}}\quad
 We now look at the solution of the main equation~\eqref{limit}.
  The equation
  reduces to
  \begin{align}\label{t}
    \left[\gamma_0(i)^2-\gamma_1(i)^2\right]\cdot T^2(\lambda_i)-\left[\gamma_1(i)^2+2y\sigma^2\gamma_0(i)\right]\cdot T(\lambda_i)+\sigma^4y^2=0~.
  \end{align}
  Since the part $\gamma_0(i)^2-\gamma_1(i)^2>0$ and $\gamma_1(i)^2+2y\sigma^2\gamma_0(i)>0$, equation \eqref{t} has two positive
  roots
  \begin{equation}\label{t1t2}
    \left\{
      \begin{array}{l}
        T_1(i)=\frac{2y\sigma^2\gamma_0(i)+\gamma_1(i)^2-\sqrt{(2y\sigma^2\gamma_0(i)+\gamma_1(i)^2)^2-4 y^2\sigma^4(\gamma_0(i)^2-\gamma_1(i)^2)}}{2\gamma_0(i)^2-2\gamma_1(i)^2}\\[0.5cm]
        T_2(i)=\frac{2y\sigma^2\gamma_0(i)+\gamma_1(i)^2+\sqrt{(2y\sigma^2\gamma_0(i)+\gamma_1(i)^2)^2-4y^2\sigma^4(\gamma_0 (i)^2-\gamma_1(i)^2)}}{2\gamma_0(i)^2-2\gamma_1(i)^2}~.
      \end{array}
    \right.
  \end{equation}
  Recall the definition of the $T$-transform that: \[T(z)=\int \frac{t}{z-t}d\mu(t)~,\]
  taking derivatives with respective to $z$ on both side leads to
  \[
  T^{'}(z)=-\int \frac{t}{(z-t)^2}d\mu(t)<0~.
  \]
  So, between the two solutions $T_1(i)$ and $T_2(i)$, only $T_1(i)$ satisfies this condition.
  And due to the fact that $\lambda_i>b$, the region of $T(\lambda_i)$
  is $(0,T(b^{+}))$, therefore the condition that there exists a unique
  solution in the region of $(0,T(b^{+}))$ is that $T_1(i) \in (0, T(b^{+}))$.

  \medskip
  The proof of the theorem is complete.
\end{proof}

\bl{
\begin{remark}
  The normal assumption in
  Theorem~\ref{mainth} is used to reduce an arbitrary
  loading matrix $A$ satisfying $A'A=I_k$ to its canonical form as
  explained in Step 2 of  the proof.
  If the loading matrix is assumed to have the canonical form, this
  normal assumption is no more necessary.
\end{remark}
}
\subsection{\bl{On the phase transition condition $T_1(i)<T(b^+)$}}
\label{condition}

\bl{
In this section, we detail the
phase transition condition $T_1(i)<T(b^+)$
that defines the detection frontier of the factors.
Unlike similar phenomenon observed for large sample covariance
matrices
as exposed in  \cite{BaikSilv06} and  \cite{BaiYao12},
this transition condition for autocovariance matrix
has a more complex nature involving the three parameters:
the limiting ratio $y$ and
the two signal-to-noise ratios (SNR)
$\gamma_0(i)/\sigma^2$ and
$\gamma_1(i)/\sigma^2$  involving the variance and lag-1 autocovariance
of the $i$-th factor time series  $(x_{it})$.
}

To start with, we observe that  the condition  \re{can be reduced to}
  \begin{align*}
  &2y\frac{\gamma_0(i)}{\sigma^2}+\left(\frac{\gamma_1(i)}{\sigma^2}\right)^2-\left(2\left(\frac{\gamma_0(i)}{\sigma^2}\right)^2-2\left(\frac{\gamma_1(i)}{\sigma^2}\right)^2\right)T(b^{+})\\
  &<\sqrt{\left(2y\frac{\gamma_0(i)}{\sigma^2}+\left(\frac{\gamma_1(i)}{\sigma^2}\right)^2\right)^2-4 y^2\left(\left(\frac{\gamma_0(i)}{\sigma^2}\right)^2-\left(\frac{\gamma_1(i)}{\sigma^2}\right)^2\right)}~,
  \end{align*}
  which has two possibilities as follows:
\begin{align}\label{cond1}
\left\{
  \begin{array}{ll}
    2y\frac{\gamma_0(i)}{\sigma^2}+\left(\frac{\gamma_1(i)}{\sigma^2}\right)^2-\left(2\left(\frac{\gamma_0(i)}{\sigma^2}\right)^2-2\left(\frac{\gamma_1(i)}{\sigma^2}\right)^2\right)T(b^{+})>0  \\
    \left(\frac{\gamma_0(i)}{\sigma^2}T(b^{+})-y\right)^2<\left(\frac{\gamma_1(i)}{\sigma^2}\right)^2(T^2(b^{+})+T(b^{+}))~,
  \end{array}
\right.
\end{align}
or
\begin{align}\label{cond2}
2y\frac{\gamma_0(i)}{\sigma^2}+\left(\frac{\gamma_1(i)}{\sigma^2}\right)^2-\left(2\left(\frac{\gamma_0(i)}{\sigma^2}\right)^2-2\left(\frac{\gamma_1(i)}{\sigma^2}\right)^2\right)T(b^{+}) \leq 0 ~.
\end{align}
\begin{figure}[h!]
\centering
\includegraphics[width=8cm]{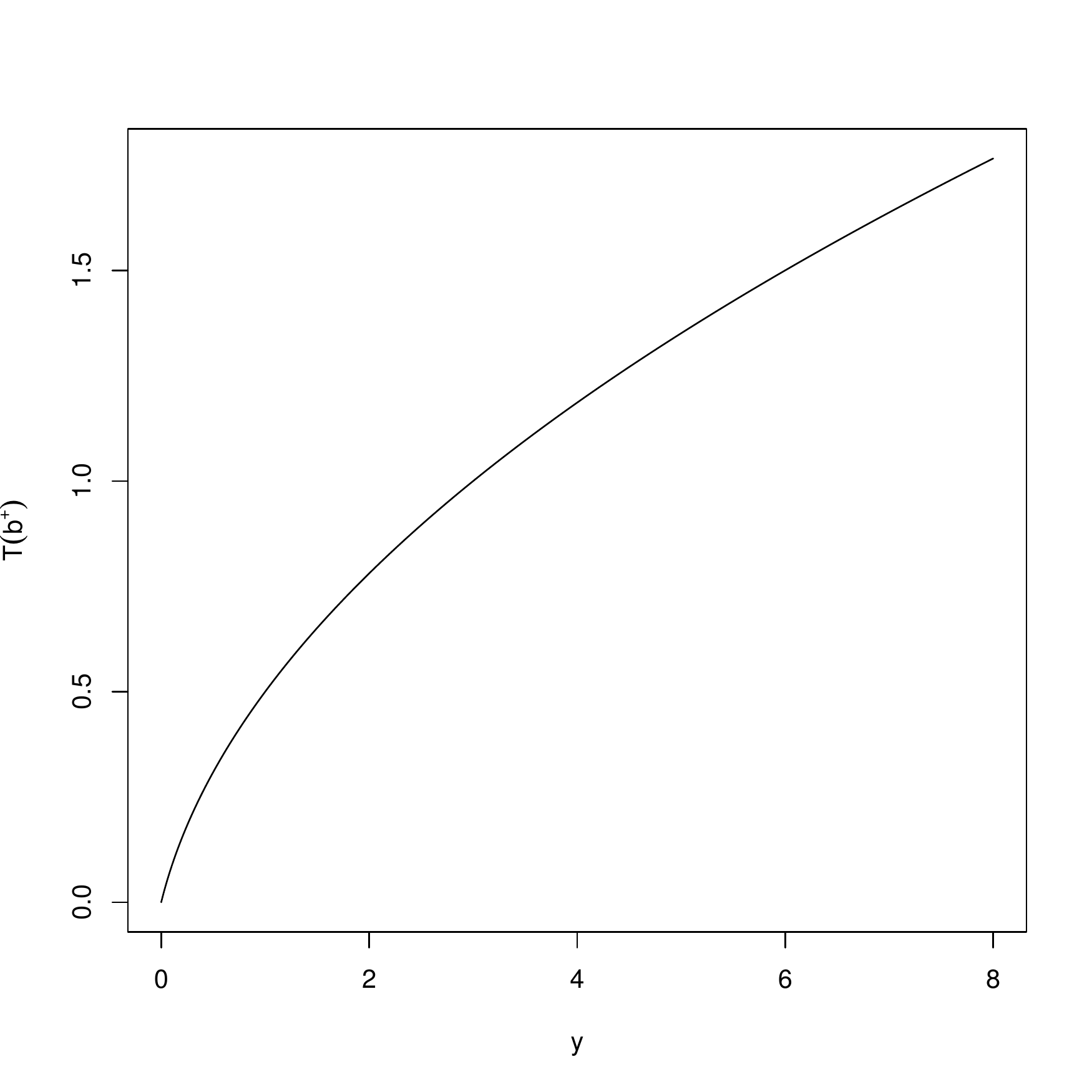}
\caption{\small Values of $T(b^{+})$ as a function of \bl{the limiting
  ratio} $y$. }\label{tb}
\end{figure}
First, we see the value of $T(b^{+})$ can be derived using
\eqref{ttr}, with the value of $b$ given in \eqref{b} as a function of
$y$, which is presented in Figure \ref{tb}. When $y$ increases from
zero to infinity, the value of $T(b^{+})$ also increases from zero to
infinity. \bl{Observe also that the slope at the origin is infinity:
$\lim_{y\to 0_+} T(b^+)/y=\infty$.}

 Once $p$ and $T$ are given ($y$ is fixed), the value of $T(b^{+})$ is fixed, then the conditions \eqref{cond1} and \eqref{cond2} can be considered as the restriction of the two parameters  $\gamma_0(i)/\sigma^2$ and $\gamma_1(i)/\sigma^2$.
And this defines a complex region in the
  $\gamma_0/\sigma^2-\gamma_1/\sigma^2$ plan which is depicted
  in  Figure \ref{region}.
The dashed curve  in  Figure \ref{region} stands for the equality
 \[
2y\frac{\gamma_0(i)}{\sigma^2}+\left(\frac{\gamma_1(i)}{\sigma^2}\right)^2-\left(2\left(\frac{\gamma_0(i)}{\sigma^2}\right)^2-2\left(\frac{\gamma_1(i)}{\sigma^2}\right)^2\right)T(b^{+}) = 0~,
\]
and the area inside this curve (the darker region) is the condition
\eqref{cond2}, while outside (the lighter region) stands for condition
\eqref{cond1}. The dotted  lines  stand for
\[
\left(\frac{\gamma_0(i)}{\sigma^2}T(b^{+})-y\right)^2=\left(\frac{\gamma_1(i)}{\sigma^2}\right)^2(T^2(b^{+})+T(b^{+}))~,
\]
and the upper and lower boundaries in solid lines   are due  to the fact that we
have always $|\gamma_1(i)| \leq \gamma_0(i)$  (by  Cauchy-Schwarz inequality).
\re{These solid and
  dotted   lines} intersect with each other at points 
$A=(\tau_0,\tau_0)$  and $B=(\tau_0,-\tau_0)$   where 
\begin{equation}
  \label{tau0}
  \tau_0 = \frac{y}{T(b^{+})+\sqrt{T^2(b^{+})+T(b^{+})}}.
\end{equation}
\bl{In other  words}, we have except for the quadrilateral region
$(*)$, \re{our conditions \eqref{cond1} and \eqref{cond1} will hold
  true, which means that }  
the corresponding factors are significant (and thus asymptotically
detectable). The   quadrilateral region
$(*)$ thus defines the phase transition boundary for the significance
of the factors. 

\begin{figure}[htb]
\centering
\mbox{  \quad $  \quad$}\\[-5cm]
\includegraphics[width=8.5cm]{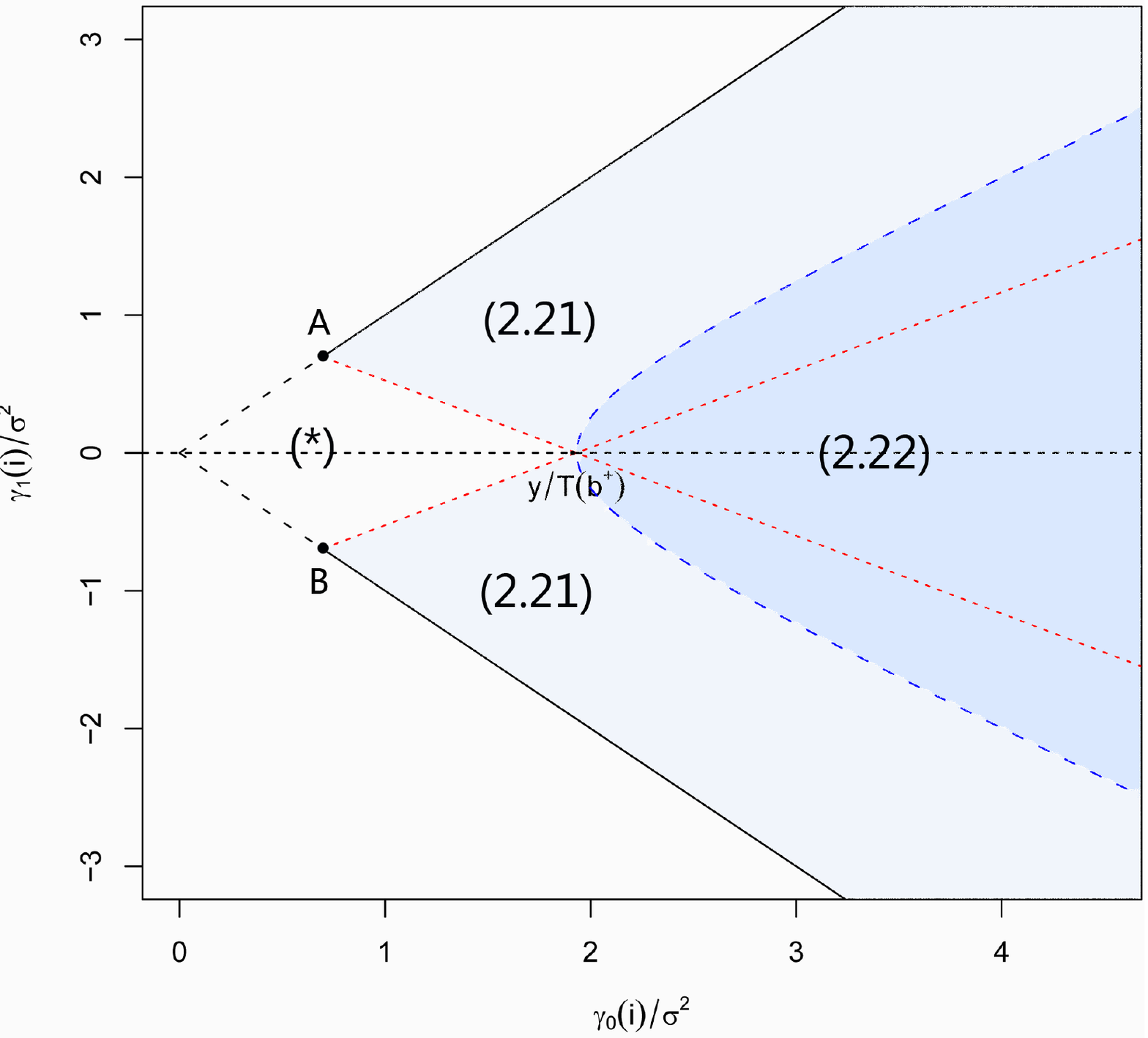}
\caption{\small Region of $\gamma_0(i)/\sigma^2$ and
  $\gamma_1(i)/\sigma^2$ that will lead to \bl{significant factors}.\label{region}}
\end{figure}

We summarize the above findings as follows.
\bl{
  \begin{corollary}\label{coro1}
    Under the same conditions as in Theorem~\ref{mainth},
    the $i$-th time series $(x_{it})$ will generate
    a significant factor in the sense that
    $T_1(i)<T(b^+)$
    if
    and only if
    either
    \begin{equation}
      \label{spike1}
      \displaystyle \frac{|\gamma_1(i)|}{\sigma^2}  \bl{>}
      \tau_0,
    \end{equation}
    or
    \begin{equation}
      \label{spike2}
      \frac{|\gamma_1(i)|}{\sigma^2} \bl{\le} \tau_0
      \quad \text{and}\quad
      \frac{\gamma_0(i)}{\sigma^2}>\frac{y-\sqrt{T^2(b^{+})+T(b^{+})}|\gamma_1(i)|/\sigma^2}{T(b^{+})},
    \end{equation}
    where the constant $\tau_0$ is given in \eqref{tau0}.
  \end{corollary}

  We now introduce  some important comments on the meaning of these
  conditions.
  \begin{enumerate}
  \item
    The essential message from these conditions is that
    {\em the $i$-th factor time  series
      is a significant factor once its strength
      $\gamma_0(i)$, or more exactly, its SNR $\gamma_0(i)/\sigma^2$  exceeds a certain level $\tau$.}
    A sufficient value for this level $\tau$ is $ \tau_1=y/ T(b^+)$ as
    shown in Figure~\ref{region}. Meanwhile,
    the SNR  should  at least equal to $\tau_0$ given in  \eqref{tau0},
    see Point A on the figure which has coordinates $(\tau_0,\tau_0)$.
    When  $\tau_0< \gamma_0(i)/\sigma^2 \le \tau_1$, the exact condition also
    depends on the lag-1 SNR  $|\gamma_1(i)|/\sigma^2$ as given in
    Eqs.~\eqref{spike1}-\eqref{spike2}.

    This is thus much in line with what is known for the phase
    transition phenomenon for large sample covariance matrices
    as exposed in \cite{BaikSilv06} and  \cite{BaiYao12}.
  \item
    As said in Introduction, in most of existing literature on
    high-dimensional factor models, the factor strengths are assumed
    to grow to infinity with the dimension $p$. Clearly,
    such {\em growing factors} are highly significant
    in our scheme, i.e. $k_0=k$,  since they will exceed the upper limit $\tau_1$ very
    quickly as the dimension $p$ grows.
  \item
    Assume that $y \to 0_+$, i.e.
    the sample size $T$ is much larger than the dimension $p$.
    Then it can be checked
    that both the quantities $\tau_0$ and $\tau_1$ will vanish.
    Therefore, when $p/T$ is small
    enough,
    any factor time series will generate
    a significant sample factor eigenvalue.
    In other words,  we have recovered
    the classical
    low-dimensional situation where $p$ is hold fixed and
    $T\to\infty$  for which
    all the $k$ factor time series can be consistently detected and
    identified.
\end{enumerate}
}
\section{Estimation of the number of factors}\label{estimation}

Let $l_1, \cdots, l_p$ be the eigenvalues of
$\widehat M  =\widehat{\Sigma}_y{ }\hat{\Sigma}_y{ }'$, sorted in decreasing order.
Assume that among the $k$ factors, the first
$k_0$ are significant which satisfy the phase transition condition
$T_1(i) < T(b^+)$, see Eq.\eqref{k0}.
Following   Theorem \ref{mainth}, the $k$  largest sample
eigenvalue $(l_i/\sigma^4)_{1\le i\le k}$ converges respectively to a limit $(\lambda_i)$,
which is larger than the right edge  $b$ of the
limiting spectral distribution
\bl{for  $1\le i\le k_0$, and equal to $b$ for $k_0 <  i\le k$.
}

It will be proved below that the
largest noise sample eigenvalues of a
given finite
number all converge to $b$, i.e.
for any fixed range $m>0$,
\begin{equation}
  \label{following}
  l_{k+1}/\sigma^4 \to b, ~\ldots~,
  l_{k+m}/\sigma^4 \to b, \quad \text{almost surely}.
\end{equation}
Consider the sequence of ratios
\begin{equation}
  \label{ratios}
  \re{\theta_j:=\frac{l_{j+1}/\sigma^4}{l_j/\sigma^4}=\frac{l_{j+1}}{l_j},\quad j\ge 1.}
\end{equation}
By definition $\theta_j\le 1$.
Therefore,
we have almost surely,
\bl{
  \begin{align}
    &\theta_j \rightarrow \frac{\lambda_{j+1}}{\lambda_j}<1~,~~~~j<k_0~,\nonumber\\
    &\theta_{k_0} \rightarrow \frac{b}{\lambda_{k_0}/\sigma^4}<1~,\label{thetaj}\\
    &
    \theta_{j}
    \rightarrow \frac{b}{b}=1~,~~~~k_0< j\le k~,\nonumber\\
    &\theta_{k+1},~\ldots,~ \theta_{k+m}  \rightarrow \frac bb=1~,
    \quad \text{for all fixed ~} m.\nonumber
  \end{align}
}

\begin{remark}
\noindent
Note that the value of $\theta_j$ is independent of $\sigma^2$. In
other words, we do not need \re{the true value of }$\sigma^2$ for
estimating the number of factors, indeed. 
\end{remark}
Let  $0< d_T<1$  be a positive constant
and
we introduce the following
estimator for the number of factors $k$:
\begin{align}\label{hatk}
  \hat{k}=\{\text{first}~ j\ge 1~ ~\text{such that }~  \theta_{j} > 1-d_T\}-1~.
\end{align}

\bl{
\begin{theorem}\label{khat}
  Consider  the  factor model~\eqref{model}
  and assume that the same conditions as in Theorem~\ref{mainth}
  are satisfied.
  Let $k_0$ be the number of significant factors defined in
  Eq.~\eqref{k0}
  and
  a
  threshold constant $d_T$ be  chosen such that
  \begin{equation}
    \label{dT}
    \max_{1\leq j \leq k_0}\lambda_{j+1}/\lambda_j<1-d_T<1.
  \end{equation}
  Then
  $\hat{k}\xrightarrow{a.s} k_0$.
\end{theorem}

This theorem thus formally establishes the fact that
the ratio estimator $\hat k$ is able to detect
all the significant factors that satisfy the
phase transition condition
given in Theorem~\ref{mainth} and detailed in
Eqs.\eqref{spike1}-\eqref{spike2}.
}

\begin{proof}(of Theorem~\ref{khat}) \quad
  As $\theta_j \xrightarrow{a.s.} \lambda_{j+1}/\lambda_j$ for $1\le
  j\le \bl{k_0}$ and by assumption \eqref{dT},
  almost surely, it will happen eventually that
  $\hat k > \bl{k_0}$.
  Next, under the claim \eqref{following} and following the limits
  given in \eqref{thetaj},
  \begin{equation}
    \label{claim}
    \theta_j \xrightarrow{a.s.} 1, \quad \text{for } j>\bl{k_0}.
  \end{equation}
  Consequently, almost surely we will have eventually $\hat k \le
  \bl{k_0}$ which, combined with the conclusion above,  proves the
  almost sure convergence of $\hat k$ to $\bl{k_0}$.

  It remains to prove the claim \eqref{following}. Since $\theta_j$ is independent of the choice of $\sigma^2$, we can assume w.l.o.g that $\sigma^2=1$ as before.  Recall that in the proof of
  Theorem~\ref{mainth}, it has been proved in
  Eqs.\eqref{yy}-\eqref{matrix1} that  if $l$ is a eigenvalue of
  $\widehat M$, then $\sqrt l$ is a  positive eigenvalue of
  the matrix
  \[
  \Gamma
  =
  \left(
    \begin{array}{cccc}
      0 &  X_1'X_0 & 0  & X_1'E_1 \\
      X_0'X_1 & 0 & X_0' E_2  & 0 \\
      0  & E_2'X_0 & 0 & E_2'E_1 \\
      E_1'X_1 & 0 & E_1'E_2   & 0
    \end{array}
  \right)~,
  \]
  which is obtained after permutation of
  the second and third row block and column block in \eqref{matrix1}
  without modifying the eigenvalues.
  Now $\Gamma$ is a symmetric block matrix
  and the positive eigenvalues of the lower diagonal block
  \[
  \left(
    \begin{array}{cc}
      0 & E_2'E_1 \\
      E_1'E_2 &  0
    \end{array}
  \right)~,
  \]
  are associated to the eigenvalues of the matrix
  $DD'=E_2'E_1E_1'E_2$  which is of dimension $p-k$
  (for the definition of these matrices, see that
  proof).
  Let $\beta_1\ge \cdots \ge \beta_{p-k}$ be the eigenvalues of
  $DD'$.
  By Cauchy interlacing theorem,
  we have
  \[      \beta_{k+1}\le  l_{k+1} \le \beta_1.
  \]
  Observing that $D$ is distributed as $\hat\Sigma_\veps{ }$ except
  that the dimension is changed from $p$ to $p-k$.  Therefore,
  the global limit of the eigenvalues of $DD'$ are the same as for the
  matrix \re{$\hat\Sigma_\veps{ }\hat\Sigma_\veps{ }'$}; in particular,
  according to Corollary~\ref{coro},
  both $\beta_{k+1}$ and $\beta_1$ converge to $b$ almost surely.
  This proves the fact that $l_{k+1}\xrightarrow{a.s} b$.
  Using similar arguments, we can establish the same fact for
  $l_{k+j}\xrightarrow{a.s} b$ for any fixed index $j\ge 1$.
  The claim \eqref{following} is thus established.
  \end{proof}


\subsection{Calibration of the tuning  parameter  $d_T$}\label{choosedn}

For
the estimator $\hat k$ in \eqref{hatk} to be practically useful, we need
to set up an appropriate value of the tuning parameter $d_T$.
\bl{Although in theory, any vanishing sequence $d_T\to 0$ will guarantee the
consistence of $\hat k$, it is preferable to have
an indicated and  practically useful sequence $(d_T)$ for real-life
data analysis.}
Here we propose an {\it  a priori}  calibration of $d_T$ based on some
knowledge from random matrix theory on the largest eigenvalues of
sample covariance matrices and of their perturbed versions.
The most important property we will use is  that
according to
such recent results on  finite rank perturbations of symmetric random
matrices, see e.g. \citet{BGM11}
it is very likely that the asymptotic distribution of
$T^{\frac{2}{3}}\lb\dfrac{l_{k+2}}{l_{k+1}}-1\rb$ 
is the same as that of
$T^{\frac{2}{3}}\lb\dfrac{\nu_{2}}{\nu_{1}}-1\rb$,
where $\nu_1$, $\nu_2$ are the two largest eigenvalues of
the base noise  matrix $\hat M_{\veps}$.
Using this similarity, we calibrate $d_T$ by simulation:  for any given
pair $(p,T)$, the distribution of
$T^{\frac{2}{3}}\lb\dfrac{\nu_{2}}{\nu_{1}}-1\rb$
is sampled using a large number \bl{(in fact 2000)}
of independent replications of standard Gaussian vectors
$\veps_t \sim N({\bf 0},{\bf  I}_p)$ and its lower 0.5\% quantile
$q_{p,T,0.5\%}$ is obtained (notice that the quantile is negative).  Using the approximation
\[ P \left\{ T^{\frac{2}{3}}\lb\dfrac{l_{k+2}}{l_{k+1}}-1\rb \le  q_{p,T,0.5\%} \right\}
\simeq
P \left\{ T^{\frac{2}{3}}\lb\dfrac{\nu_{2}}{\nu_{1}}-1\rb \le  q_{p,T,0.5\%} \right\}=0.5\%,
\]
we calibrate $d_T$ at the value  $d_T =
|q_{p,T,0.5\%}|/T^{2/3}$. Notice that $d_T$ vanishes at rate
$T^{-2/3}$.
In all the simulation experiments in Section~\ref{simulation}  or for the data
analysis reported in Section~\ref{application}, this tuned value of
$d_T$ is used for the  given pairs   $(T,p)$.

\section{Monte-Carlo experiments}
\label{simulation}

In this section, we report some simulation results to
\re{show} the finite-sample performance of our estimator. \re{For the reason of} robustness, we will consider a
reinforced estimator
$\hat{k}^*$ defined as
\begin{align}\label{hattk}
  \hat{k}^{*}=\{
  \text{first}~ j\ge 1~ ~\text{such that }~  \theta_{j} > 1-d_T
  ~\text{and }~  \theta_{j+1} > 1-d_T\}-1~.
\end{align}
Clearly, $\hat k^*$ is asymptotically equivalent to the initial
estimator  $\hat k$ which uses only \re{one} single
test value $j$.
As for the factor model,
we adopt the same settings as in \cite{LY012} where
\[y_t=A x_t+\veps_t,~\veps_t\sim N_p({\bf 0},{\bf I}_p),\]
\[x_t=\Theta x_{t-1}+e_t,~ e_t\sim N_k({\bf 0},\Gamma),\]
\noindent
where $A$ is a $p\times k$ matrix, w.l.o.g, we set the variance $\sigma^2$ of the white noise $\varepsilon_t$ to be 1.

In \cite{LY012}, the factor loading matrix $A$ are independently
generated from uniform distribution on the interval $[-1,1]$ first and
then divided by $p^{\delta/2}$ where $\delta\in [0,1]$. The induced
$k$ factor strengths are thus of order $O(p^{1-\delta})$. Their
estimator of number of factors is recalled in \eqref{eq:LY}. Cases
\bl{where
  three factors are either all very strong with $\delta=0$ or all
  moderately strong   with $\delta=0.5$  are  discussed in details in
  that paper.}
The results show that $\tilde{k}$ performs better when
factors are stronger. An experimental setting with a combination of
two strong factors and one moderate factor indicates that a two-step
estimation procedure needs to be employed in order to identify all
three factors. In each step only factors with the highest level of
strength can be detected.

While in our case, the coefficient matrix $A$ satisfies $A'A=I_k$. Considering the eigenvalues of $\widehat M$ are invariant under orthogonal transformation (See Step 2 in the proof of Theorem \ref{mainth}), we fix
 \[
  A=\begin{pmatrix}
    I_k \\
    \bm{0}_{p-k} \\
  \end{pmatrix}~.
  \]
  \noindent
 Then we manipulate the factor strength by adjusting the value of $\Theta$ and $\Gamma$. To ensure the stationarity of $\{y_t\}$ process and the Independence among the components of the factor process $\{x_t\}$, $\Theta$ and $\Gamma$ are both diagonal matrices and the diagonal elements of $\Theta$ lie within $(-1,1)$. To keep pace with the settings in \cite{LY012}, we multiply $p^{\frac{1-\delta}{2}}$ with the diagonal entries of $\Gamma$ to adjust the corresponding factor strength. It can be seen that when $\delta=0$, the factor is strongest while with $\delta=1$, the factor is weakest.

The entire simulation study is mainly composed of four parts formulated in four different scenarios as follows:
\begin{itemize}
\item[(I)] Two very strong factors with $\delta_1=0.5$ and $\delta_2=0.8$
  \bl{and}
 \[~\Theta=\lb
  \begin{array}{cc}
    0.6&0\\
    0&0.5
  \end{array}
  \rb,~ \Gamma=\lb
  \begin{array}{cc}
    4\times p^{\frac{1-\delta_1}{2}}&0\\
    0&4\times p^{\frac{1-\delta_2}{2}}
  \end{array}
  \rb\bl{.}
  \]

\item[(II)]  Four weak factors with same strength level $\delta=1$;  three of them
  are significant with their theoretical limits
  $\lambda_1,~\lambda_2,~\lambda_3$ all keeping a  moderate distance
  from $b$  while  the fourth factor is insignificant with 
  its theoretical limit $\lambda_4$ equal to right edge  $b$ of the
  noise eigenvalues. Precisely, 
  \[~\Theta=\lb
  \begin{array}{cccc}
    0.6&0&0&0\\
    0&-0.5&0&0\\
    0&0&0.3&0\\
    0&0&0&0.2
  \end{array}
  \rb,~ \Gamma=\lb
  \begin{array}{cccc}
    4&0&0&0\\
    0&4&0&0\\
    0&0&4&0\\
    0&0&0&1
  \end{array}
  \rb\bl{.}\]
\item[(III)] Three weak factors with $\delta=1$ and $\lambda_3$ \re{stays} very close to $b$
  \bl{and}
  \[~\Theta=\lb
  \begin{array}{ccc}
    0.6&0&0\\
    0&-0.5&0\\
    0&0&0.3
  \end{array}
  \rb,~ \Gamma=\lb
  \begin{array}{ccc}
    2&0&0\\
    0&2&0\\
    0&0&2
  \end{array}
  \rb.\]
\item[(IV)] A \bl{mixed}  case with two strong factors with
  $\delta_1=0.5,~\delta_2=0.8$, and five weak factors with $\delta=1$,
  and
  \[\Theta=diag(0.6,0.5,0.6,-0.5,0.3,0.6,-0.5),\]
  \[\Gamma=diag(4\times p^{\frac{1-\delta_1}{2}},4\times p^{\frac{1-\delta_2}{2}},4,4,4,2,2).\]
\end{itemize}

Recall that for the estimator $\hat{k}^*$, the critical value $d_T$ is
calibrated as explained in Section \ref{choosedn} using the simulated empirical
0.5\% lower quantile. We set $p=100,~300,~500,~1000,~1500$,
$T=0.5p,~2p$, i.e $y=2,~0.5$.
\bl{It will be seen below that in general, the cases with  $T=0.5p$
  will be harder to deal with than the cases with  $T=2p$.
}
We repeat 1000 times to calculate the
\bl{empirical  frequencies  of the different decisions
  $(\hat{k}^*=k_0)$, $(\hat{k}^*=k_0\pm1)$ and $(|\hat{k}^*-k_0|>1)$.}
The results are as follows.

\begin{itemize}
\item[(I)] In Scenario I, we have two very strong factors with
  $\delta_1=0.5$ and $\delta_2=0.8$
  and their  strengths grow to infinity with $p$.
  Thus $k_0=k=2$ and the  two factors  must be  easily detectable.
  As seen from Table~\ref{scena1}, our estimator $\hat{k}^*$ converges very fast
  to the true number of factors. On the other hand, the one-step
  estimator \re{$\tilde{k}$} of \citet{LY012}
  tends to detect only one factor in each step due to the fact that the two factors are \re{of} different strength.

{\small
\begin{table}[!h]
\caption{Scenario I with two strong factors ($k_0=k=2$)\label{scena1}}
\begin{tabular}{+c^c^c^c^c^c|^c^c^c^c^c^c}
\hline
$p$ & 100 & 300 & 500 & 1000 & 1500 & $p$ & 100 & 300 & 500 & 1000 & 1500\tabularnewline
$T=2p$ & 200 & 600 & 1000 & 2000 & 3000 & $T=2p$ & 200 & 600 & 1000 & 2000 & 3000\tabularnewline
\hline
$\tilde{k}=1$ & 0.343 & 0.294 & 0.257 & 0.287 & 0.317 &
$\hat{k}^{*}=1$ & 0 & 0 & 0 & 0 & 0\tabularnewline
\rowstyle{\bfseries}%
$\tilde{k}=k_0$ & 0.657 & 0.706 & 0.743 & 0.713 & 0.683 & $\hat{k}^{*}=k_0$ & 0.974 & 0.984 & 0.993 & 0.996 & 0.998\tabularnewline
$\tilde{k}\geq3$ & 0 & 0 & 0 & 0 & 0 & $\hat{k}^{*}\geq3$ & 0.026 & 0.016 & 0.007 & 0.004 & 0.002\tabularnewline
\hline
$p$ & 100 & 300 & 500 & 1000 & 1500 & $p$ & 100 & 300 & 500 & 1000 & 1500\tabularnewline
$T=0.5p$ & 50 & 150 & 250 & 500 & 750 & $T=0.5p$ & 50 & 150 & 250 & 500 & 750\tabularnewline
\hline
$\tilde{k}=1$ & 0.786 & 0.801 & 0.876 & 0.96 & 0.992 & $\hat{k}^{*}=1$ & 0.086 & 0 & 0 & 0 & 0\tabularnewline
\rowstyle{\bfseries}%
$\tilde{k}=k_0$ & 0.21 & 0.199 & 0.124 & 0.04 & 0.008 & $\hat{k}^{*}=k_0$ & 0.771 & 0.882 & 0.896 & 0.881 & 0.881\tabularnewline
$\tilde{k}\geq3$ & 0.004 & 0 & 0 & 0 & 0 & $\hat{k}^{*}\geq3$ & 0.143 & 0.118 & 0.104 & 0.119 & 0.119\tabularnewline
\hline
\end{tabular}
\end{table}
}
\item [(II)] In Scenario II, we have four weak factors of same
  strength level $\delta=1$.
  \bl{
    The  theoretical limits related to Theorem \ref{mainth} are
    displayed in
    Table~\ref{scena2-lim}. 
    Figure~\ref{table2y05} for $T=2p$ and Figure~\ref{table2y2} for
    $T=0.5p$  depict the
    position of these four factors (numbered from 1 to 4)
    in the phase transition region 
    defined in Corollary~\ref{coro1} and we see 
    three among the four lying  inside the detectable area in both situations.
    It can be seen from the table that  for  both
    combinations of $T=2p$ and $T=0.5p$,  the first three limits
    $\lambda_i$ are far from upper bound $b$ and the fourth limit $\lambda_4$ equals to $b$.
    We thus have three significant factors
    ($k_0=3$)
    which are  detectable while the fourth one is too weak for the
    detection.
    Results in Table~\ref{scena2} show that both the estimators
    $\tilde k$ (one-step)  and $\hat k^*$ are consistent with however a much
    higher convergence speed for $\hat k^*$.
  }

{\small
\begin{table}[!h]
\caption{Scenario II -  Theoretical limits ($k_0=3,~k=4$)\label{scena2-lim}}
\centering
\begin{tabular}{+c|^c^c^c^c|^c^c^c^c|^c^c^c^c}
\hline
 &&  &  &  & \multicolumn{4}{c|}{$T=2p$} & \multicolumn{4}{c}{$T=0.5p$}\tabularnewline
\cline{6-13}
No.&$\Theta$ & $\Gamma$ & $\gamma_{0}\left(i\right)$ & $\gamma_{1}\left(i\right)$ & $T_{1}(i)$ & $T(b^+)$ & $\lambda_{i}$ & $b$ & $T_{1}(i)$ & $T(b^+)$  & $\lambda_{i}$ & $b$\tabularnewline
\hline
(1)&0.6 & 4 & 6.25 & 3.75 & 0.0125 &0.3076 & 21.2 & 2.7725 & 0.1102 & 0.7775 & 44.8 & 17.6366\tabularnewline
(2)&-0.5 & 4 & 5.33 & -2.67 & 0.021 &0.3076& 13.1 & 2.7725 & 0.1596 & 0.7775 & 33.85 & 17.6366\tabularnewline
(3)&0.3 & 4 & 4.3956 & 1.3187 & 0.047 &0.3076& 6.65 & 2.7725 & 0.2767 & 0.7775 & 23.92 & 17.6366\tabularnewline
(4) & 0.2 & 1 & 1.042 & 0.2083 & 0.3446 & 0.3076 & {\bf 2.7725} & {\bf
  2.7725} & 1.5296 & 0.7775 & {\bf 17.6366} &1{\bf 7.6366}\tabularnewline
\hline
\end{tabular}
\end{table}
}

\begin{figure}[!h]
  \centering
  \includegraphics[width=0.55\textwidth,trim=0.0cm 0.4cm 0.4cm 0.4cm,clip]{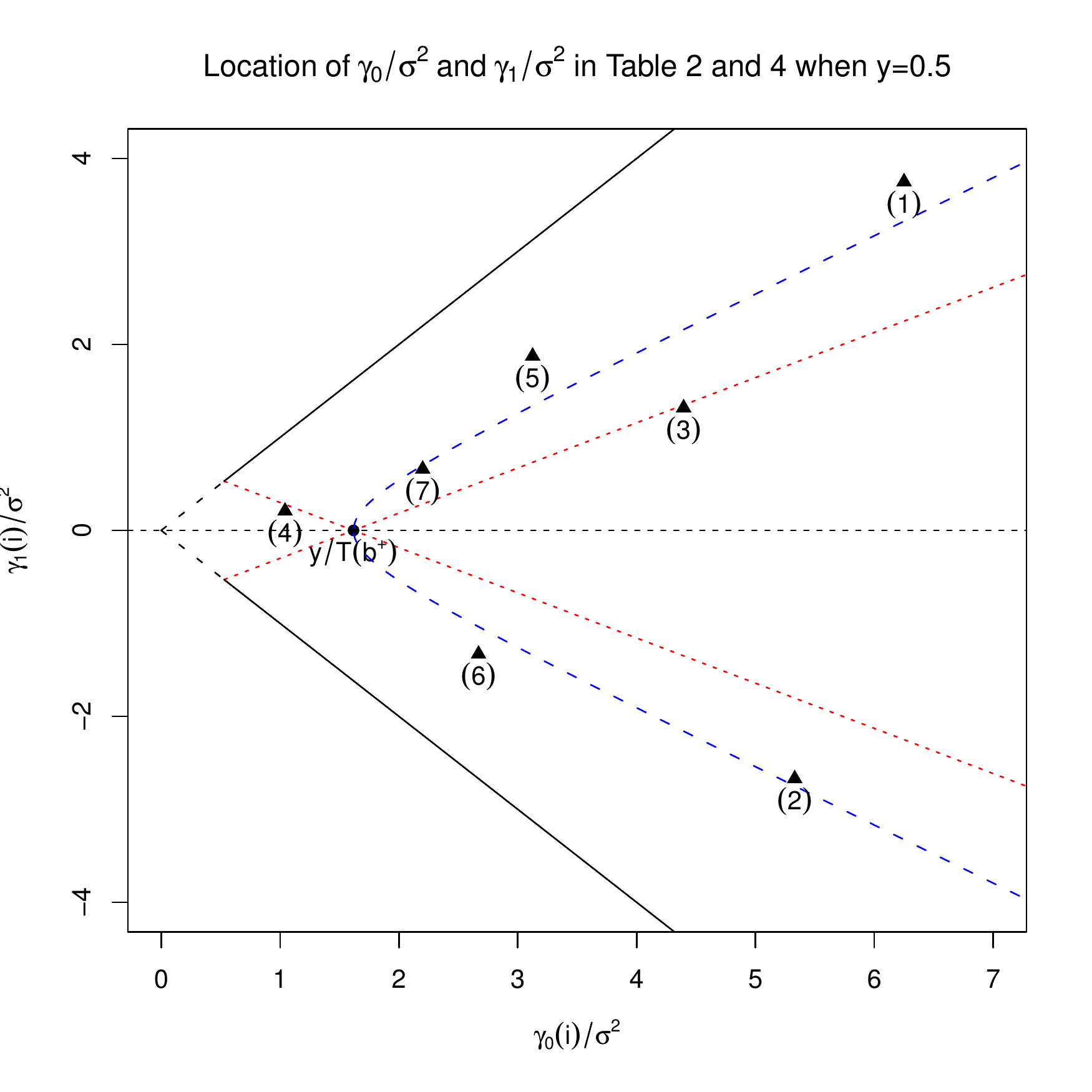}
  \caption{\small
    Locations of factor SNR's  $(\gamma_0/,\gamma_1)/\sigma^2$
    from Tables \protect\ref{scena2-lim}
    (points numbered from 1 to 4),
    \protect\ref{scena3-lim}     (points numbered from 5 to 7),
    and
    \protect\ref{scena4-lim} (points numbered 1-2-3-5-6) 
    with $T=2p$ ($y=0.5$).
    \label{table2y05}}
\end{figure}

\begin{figure}[!h]
  \centering
  \includegraphics[width=0.55\textwidth,trim=0.0cm 0.4cm 0.4cm 0.4cm,clip]{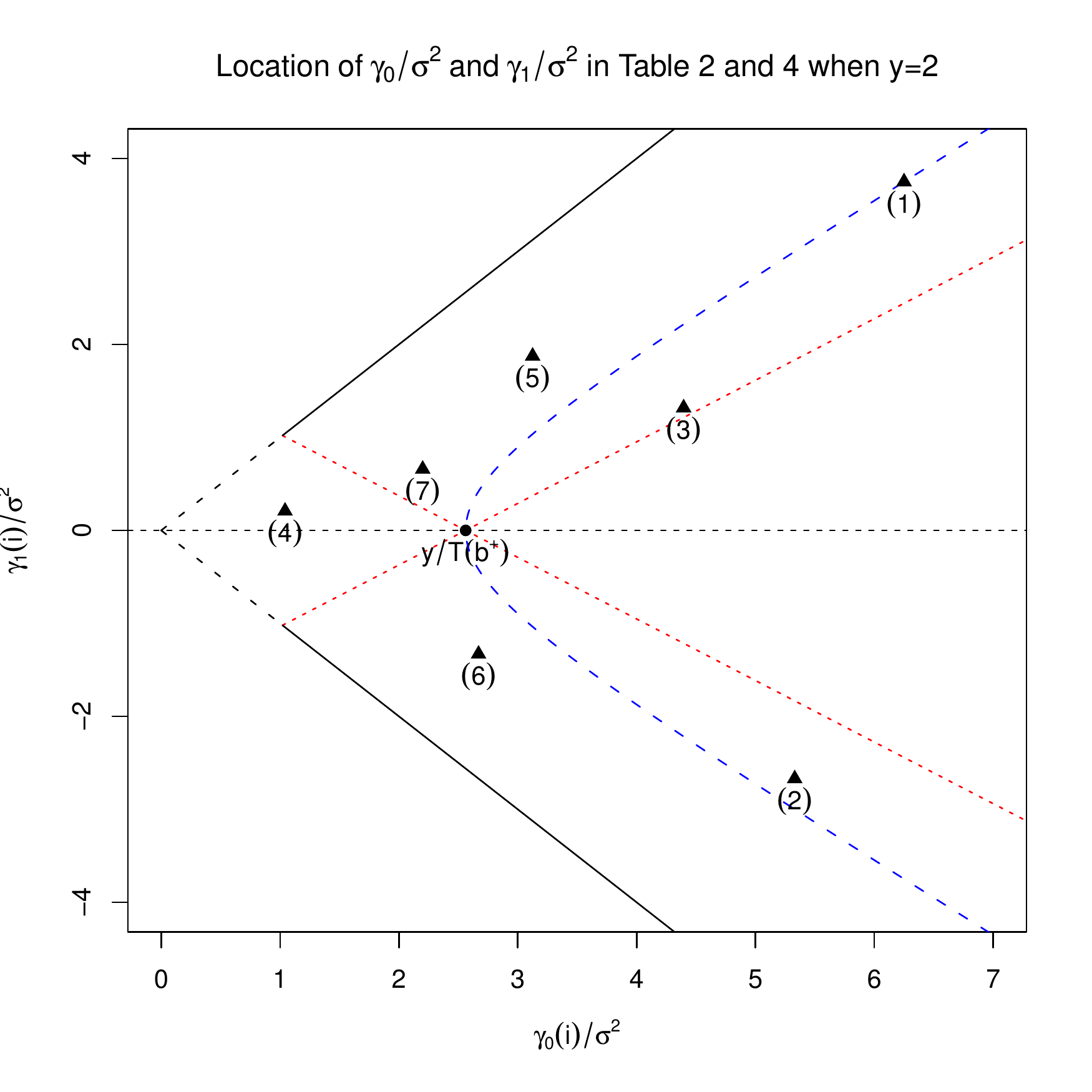}
  \caption{\small
    Locations of factor SNR's  $(\gamma_0/,\gamma_1)/\sigma^2$
    from Tables \protect\ref{scena2-lim}
    (points numbered from 1 to 4),
    \protect\ref{scena3-lim}     (points numbered from 5 to 7),
    and
    \protect\ref{scena4-lim} (points numbered 1-2-3-5-6) 
    with $T=0.5p$ ($y=2$).
    \label{table2y2}}
\end{figure}


{\small
\begin{table}
  \caption{Scenario II with three weak yet significant factors among
    four $(k_0=3, ~k=4$)\label{scena2}}
  \centering
    \resizebox{0.98\textwidth}{!}{
\begin{tabular}{+c|^c|^c|^c|^c|^c|^c|^c|^c|^c|^c|^c}
\hline
p & 100 & 300 & 500 & 1000 & 1500 & p & 100 & 300 & 500 & 1000 & 1500\tabularnewline
T=2p & 200 & 600 & 1000 & 2000 & 3000 & T=2p & 200 & 600 & 1000 & 2000 & 3000\tabularnewline
\hline
$\tilde{k}=1$ & 0.152 & 0.074 & 0.045 & 0.01 & 0.001 & $\hat{k}^{*}=1$ & 0.005 & 0 & 0 & 0 & 0\tabularnewline
$\tilde{k}=2$ & 0.402 & 0.344 & 0.276 & 0.194 & 0.126 & $\hat{k}^{*}=2$ & 0.026 & 0 & 0 & 0 & 0\tabularnewline
\rowstyle{\bfseries}%
$\tilde{k}=k_{0}$ & 0.446 & 0.582 & 0.679 & 0.796 & 0.873 & $\hat{k}^{*}=k_{0}$ & 0.928 & 0.967 & 0.953 & 0.96 & 0.966\tabularnewline
$\tilde{k}=4$ & 0 & 0 & 0 & 0 & 0 & $\hat{k}^{*}=4$ & 0.04 & 0.033 & 0.046 & 0.04 & 0.033\tabularnewline
$\tilde{k}\geq5$ & 0 & 0 & 0 & 0 & 0 & $\hat{k}^{*}\geq5$ & 0.001 & 0 & 0.001 & 0 & 0.001\tabularnewline
\hline
p & 100 & 300 & 500 & 1000 & 1500 & p & 100 & 300 & 500 & 1000 & 1500\tabularnewline
T=0.5p & 50 & 150 & 250 & 500 & 750 & T=0.5p & 50 & 150 & 250 & 500 & 750\tabularnewline
\hline
$\tilde{k}=1$ & 0.479 & 0.368 & 0.344 & 0.284 & 0.289 & $\hat{k}^{*}=1$ & 0.376 & 0.02 & 0.003 & 0 & 0\tabularnewline
$\tilde{k}=2$ & 0.406 & 0.432 & 0.454 & 0.495 & 0.514 & $\hat{k}^{*}=2$ & 0.456 & 0.221 & 0.048 & 0.001 & 0\tabularnewline
\rowstyle{\bfseries}%
$\tilde{k}=k_{0}$ & 0.105 & 0.199 & 0.202 & 0.221 & 0.197 & $\hat{k}^{*}=k_{0}$ & 0.16 & 0.73 & 0.915 & 0.986 & 0.982\tabularnewline
$\tilde{k}=4$ & 0.006 & 0.001 & 0 & 0 & 0 & $\hat{k}^{*}=4$ & 0.008 & 0.029 & 0.03 & 0.013 & 0.017\tabularnewline
$\tilde{k}\geq5$ & 0.004 & 0 & 0 & 0 & 0 & $\hat{k}^{*}\geq5$ & 0 & 0 & 0.004 & 0 & 0.001\tabularnewline
\hline
\end{tabular}}
\end{table}
}

\item [(III)] Theoretical limits and empirical result for Scenario III
  are presented in Table~\ref{scena3-lim}, 
  Figures~\ref{table2y05}
  and \ref{table2y2}, and Table~\ref{scena3}.
  For both situations of $T=0.5p$ and $T=2p$, the model has three
  significant factors ($k_0=k=3$).
  Notice however that when $T=0.5p$, the 3rd factor is quite weak and the corresponding
  limit $\lambda_3=17.95$ is very close to the right edge $b=17.64$
  so that this factor would be detectable only in theory (or with
  very large sample sizes).  This is also easily verified in
  Figure~\ref{table2y2} 
  that 
  the point (3) corresponding to the
  weakest factor lies very close to the boundary of the detectable
  region.  As for the empirical values in Table~\ref{scena3},
  the estimator $\hat k^*$ converges
  quickly when $T=2p$ and much more slowly when $T=0.5p$.
  Meanwhile, the  estimator $\tilde k$ (with one-step) seems
  inconsistent even in the easier case of $T=2p$.
 
{\small
\begin{table}[!h]
\caption{Scenario III - Theoretical limits ($k_0=k=3$)\label{scena3-lim}}
\centering
\begin{tabular}{+c|^c^c^c^c|^c^c^c^c|^c^c^c^c}
\hline
 &&  &  &  & \multicolumn{4}{c|}{$T=2p$} & \multicolumn{4}{c}{$T=0.5p$}\tabularnewline
\cline{6-13}
No. &$\Theta$ & $\Gamma$ & $\gamma_{0}\left(i\right)$ & $\gamma_{1}\left(i\right)$ & $T_{1}(i)$ & $T(b^+)$ & $\lambda_{i}$ & $b$ & $T_{1}(i)$ & $T(b^+)$ & $\lambda_{i}$ & $b$\tabularnewline
\hline
(5)&0.6 & 2 & 3.125 & 1.875 & 0.0391 &0.3076 & 7.65 & 2.7725 & 0.2845 &0.7775& 23.79 & 17.6366\tabularnewline
(6)&-0.5 & 2 & 2.67 & -1.33 & 0.0607 &0.3076& 5.48 & 2.7725 & 0.3852 &0.7775& 20.45 & 17.6366\tabularnewline
(7)&0.3 & 2 & 2.20 & 0.659 & 0.1183 &0.3076 & 3.61 & 2.7725 & 0.6116
&0.7775& {\bf 17.95} & {\bf 17.6366}\tabularnewline
\hline
\end{tabular}
\end{table}
}

{\small
\begin{table}[!h]
  \caption{Scenario III with three weak yet insignificant factors ($k_0=k=3$)\label{scena3}}
  \centering
  \resizebox{0.98\textwidth}{!}{
    \begin{tabular}{+c^c^c^c^c^c|^c^c^c^c^c^c}
      \hline
      $p$ & 100 & 300 & 500 & 1000 & 1500 & $p$ & 100 & 300 & 500 & 1000 & 1500\tabularnewline
      $T=2p$ & 200 & 600 & 1000 & 2000 & 3000 & $T=2p$ & 200 & 600 & 1000 & 2000 & 3000\tabularnewline
      \hline
      $\tilde{k}<2$ & 0.403 & 0.322 & 0.327 & 0.302 & 0.308 & $\hat{k}^{*}<2$ & 0.074 & 0 & 0 & 0 & 0\tabularnewline
      $\tilde{k}=2$ & 0.454 & 0.587 & 0.598 & 0.653 & 0.669 & $\hat{k}^{*}=2$ & 0.441 & 0.047 & 0.005 & 0 & 0\tabularnewline
\rowstyle{\bfseries}%
      $\tilde{k}=k_0$ & 0.143 & 0.091 & 0.075 & 0.045 & 0.023 & $\hat{k}^{*}=k_0$ & 0.48 & 0.945 & 0.991 & 0.996 & 0.999\tabularnewline
      $\tilde{k}\geq 4$ & 0 & 0 & 0 & 0 & 0 & $\hat{k}^{*}\geq 4$ & 0.005 & 0.008 & 0.004 & 0.004 & 0.001\tabularnewline
      \hline
      $p$ & 100 & 300 & 500 & 1000 & 1500 & $p$ & 100 & 300 & 500 & 1000 & 1500\tabularnewline
      $T=0.5p$ & 50 & 150 & 250 & 500 & 750 & $T=0.5p$ & 50 & 150 & 250 & 500 & 750\tabularnewline
      \hline
      $\tilde{k}<2$ & 0.548 & 0.57 & 0.589 & 0.548 & 0.547 & $\hat{k}^{*}<2$ & 0.886 & 0.639 & 0.435 & 0.114 & 0.049\tabularnewline
      $\tilde{k}=2$ & 0.264 & 0.359 & 0.371 & 0.437 & 0.447 & $\hat{k}^{*}=2$ & 0.107 & 0.338 & 0.508 & 0.718 & 0.745\tabularnewline
\rowstyle{\bfseries}%
      $\tilde{k}=k_0$ & 0.08 & 0.053 & 0.036 & 0.015 & 0.006 & $\hat{k}^{*}=k_0$ & 0.006 & 0.022 & 0.057 & 0.167 & 0.205\tabularnewline
      $\tilde{k}\geq 4$ & 0.108 & 0.018 & 0.004 & 0 & 0 & $\hat{k}^{*}\geq 4$ & 0.001 & 0.001 & 0 & 0.001 & 0.001\tabularnewline
      \hline
    \end{tabular}}
\end{table}
}
\item [(IV)]
  \bl{Scenario IV is the {most complex case}  with two very strong factors
    and five weak factors.
    As predicted by the theory,
    the two largest factor
    eigenvalues $l_1,~l_2$ of $\hat{M}$ blow up to infinity
    while the following  5 factor eigenvalues $l_3 \sim
    l_7$   converge to a $\lambda_i>b$.
    The corresponding theoretical limits for the five weak factors are given in
    Table~\ref{scena4-lim} and their SNR's depicted in Figures~\ref{table2y05}
    and \ref{table2y2}.
    Meanwhile,  all the  $k_0=k=7$ factors are significant.
    Clearly in this scenario, the performance of the one-step estimator
    $\tilde k$, denoted as $\tilde{k}^{(1)}$,  is quite limited and in order to make a closer
    comparison  with our estimator $\hat{k}^*$,
    we have also run  the two-step and the three-step versions of the
    estimator $\tilde k$. Among these two versions
    we report the best results obtained by the three-step version
    (denoted as $\tilde{k}^{(3)}$).
  }
  It can be seen from Table~\ref{scena4} that our estimator is able to
  detect the 7 factors with multi-level strength in a single step while
  $\tilde{k}$ can only identify one factor in each step: i.e.
  $\tilde{k}^{(1)}\rightarrow 1$ and $\tilde{k}^{(3)}\rightarrow 3$.

{\small
\begin{table}
  \caption{Scenario IV  - Theoretical limits ($k_0=k=7$)\label{scena4-lim}}
  \centering
\begin{tabular}{+c|^c^c^c^c|^c^c^c^c|^c^c^c^c}
\hline
 &  &  &  &  & \multicolumn{4}{c|}{$T=2p$} & \multicolumn{4}{c}{$T=0.5p$}\tabularnewline
\cline{6-13}
NO. & $\Theta$ & $\Gamma$ & $\gamma_{0}\left(i\right)$ & $\gamma_{1}\left(i\right)$ & $T_{1}(i)$ & $T\left(b^{+}\right)$ & $\lambda_{i}$ & $b$ & $T_{1}(i)$ & $T\left(b^{+}\right)$ & $\lambda_{i}$ & $b$\tabularnewline
\hline
(1) & 0.6 & 4 & 6.25 & 3.75 & 0.0125 & 0.3076 & 21.2 & 2.7725 & 0.1102 & 0.7775 & 44.8 & 17.6366\tabularnewline
(2) & -0.5 & 4 & 5.33 & -2.67 & 0.021 & 0.3076 & 13.1 & 2.7725 & 0.1596 & 0.7775 & 33.85 & 17.6366\tabularnewline
(3) & 0.3 & 4 & 4.3956 & 1.3187 & 0.047 & 0.3076 & 6.65 & 2.7725 & 0.2767 & 0.7775 & 23.92 & 17.6366\tabularnewline
(5) & 0.6 & 2 & 3.125 & 1.875 & 0.0391 & 0.3076 & 7.65 & 2.7725 & 0.2845 & 0.7775 & 23.79 & 17.6366\tabularnewline
(6) & -0.5 & 2 & 2.67 & -1.33 & 0.0607 & 0.3076 & 5.48 & 2.7725 & 0.3852 & 0.7775 & 20.45 & 17.6366\tabularnewline
\hline
\end{tabular}
\end{table}
}

{\small
\begin{table}
  \caption{Scenario IV with seven factors of  multiple  strength
    levels ($k_0=k=7$)\label{scena4}}
  \centering
\begin{tabular}{+c^c^c^c^c^c|^c^c^c^c^c^c}
\hline
p & 100 & 300 & 500 & 1000 & 1500 & p & 100 & 300 & 500 & 1000 & 1500\tabularnewline
T=2p & 200 & 600 & 1000 & 2000 & 3000 & T=0.5p & 50 & 150 & 250 & 500 & 750\tabularnewline
\hline
$\tilde{k}^{(1)}=1$ & 0.696 & 0.858 & 0.949 & 0.995 & 1 & $\tilde{k}^{(1)}=1$ & 0.73 & 0.812 & 0.881 & 0.95 & 0.986\tabularnewline
$\tilde{k}^{(1)}=2$ & 0.244 & 0.137 & 0.051 & 0.005 & 0 & $\tilde{k}^{(1)}=2$ & 0.211 & 0.177 & 0.118 & 0.05 & 0.014\tabularnewline
$\tilde{k}^{(1)}=3$ & 0.033 & 0.004 & 0 & 0 & 0 & $\tilde{k}^{(1)}=3$ & 0.039 & 0.011 & 0.001 & 0 & 0\tabularnewline
$\tilde{k}^{(1)}=4$ & 0.019 & 0.001 & 0 & 0 & 0 & $\tilde{k}^{(1)}=4$ & 0.015 & 0 & 0 & 0 & 0\tabularnewline
$\tilde{k}^{(1)}=5$ & 0.005 & 0 & 0 & 0 & 0 & $\tilde{k}^{(1)}=5$ & 0.004 & 0 & 0 & 0 & 0\tabularnewline
$\tilde{k}^{(1)}=6$ & 0.002 & 0 & 0 & 0 & 0 & $\tilde{k}^{(1)}=6$ & 0.001 & 0 & 0 & 0 & 0\tabularnewline
\rowstyle{\bfseries}%
$\tilde{k}^{(1)}=k_0$ & 0.001 & 0 & 0 & 0 & 0 & $\tilde{k}^{(1)}=k_0$ & 0 & 0 & 0 & 0 & 0\tabularnewline
$\tilde{k}^{(1)}\geq8$ & 0 & 0 & 0 & 0 & 0 & $\tilde{k}^{(1)}\geq8$ & 0 & 0 & 0 & 0 & 0\tabularnewline
\hline
p & 100 & 300 & 500 & 1000 & 1500 & p & 100 & 300 & 500 & 1000 & 1500\tabularnewline
T=2p & 200 & 600 & 1000 & 2000 & 3000 & T=0.5p & 50 & 150 & 250 & 500 & 750\tabularnewline
\hline
$\tilde{k}^{(3)}=1$ & 0 & 0 & 0 & 0 & 0 & $\tilde{k}^{(3)}=1$ & 0 & 0 & 0 & 0 & 0\tabularnewline
$\tilde{k}^{(3)}=2$ & 0 & 0 & 0 & 0 & 0 & $\tilde{k}^{(3)}=2$ & 0 & 0 & 0 & 0 & 0\tabularnewline
$\tilde{k}^{(3)}=3$ & 0.691 & 0.875 & 0.945 & 0.998 & 0.999 & $\tilde{k}^{(3)}=3$ & 0.71 & 0.802 & 0.862 & 0.955 & 0.982\tabularnewline
$\tilde{k}^{(3)}=4$ & 0.002 & 0 & 0 & 0 & 0 & $\tilde{k}^{(3)}=4$ & 0 & 0 & 0 & 0 & 0\tabularnewline
$\tilde{k}^{(3)}=5$ & 0 & 0 & 0 & 0 & 0 & $\tilde{k}^{(3)}=5$ & 0.001 & 0 & 0 & 0 & 0\tabularnewline
$\tilde{k}^{(3)}=6$ & 0.244 & 0.125 & 0.055 & 0.002 & 0.001 & $\tilde{k}^{(3)}=6$ & 0.212 & 0.192 & 0.135 & 0.045 & 0.018\tabularnewline
\rowstyle{\bfseries}%
$\tilde{k}^{(3)}=k_0$ & 0 & 0 & 0 & 0 & 0 & $\tilde{k}^{(3)}=k_0$ & 0 & 0 & 0 & 0 & 0\tabularnewline
$\tilde{k}^{(3)}\geq8$ & 0.063 & 0 & 0 & 0 & 0 & $\tilde{k}^{(3)}\geq8$ & 0.077 & 0.006 & 0.003 & 0 & 0\tabularnewline
\hline
p & 100 & 300 & 500 & 1000 & 1500 & p & 100 & 300 & 500 & 1000 & 1500\tabularnewline
T=2p & 200 & 600 & 1000 & 2000 & 3000 & T=0.5p & 50 & 150 & 250 & 500 & 750\tabularnewline
\hline
$\hat{k}^{*}=1$ & 0.012 & 0 & 0 & 0 & 0 & $\hat{k}^{*}=1$ & 0.151 & 0.01 & 0 & 0 & 0\tabularnewline
$\hat{k}^{*}=2$ & 0.031 & 0.001 & 0 & 0 & 0 & $\hat{k}^{*}=2$ & 0.25 & 0.038 & 0.01 & 0.001 & 0\tabularnewline
$\hat{k}^{*}=3$ & 0.034 & 0.002 & 0 & 0 & 0 & $\hat{k}^{*}=3$ & 0.28 & 0.065 & 0.027 & 0.003 & 0\tabularnewline
$\hat{k}^{*}=4$ & 0.062 & 0.015 & 0.006 & 0.001 & 0 & $\hat{k}^{*}=4$ & 0.254 & 0.227 & 0.107 & 0.022 & 0.007\tabularnewline
$\hat{k}^{*}=5$ & 0.049 & 0 & 0 & 0 & 0 & $\hat{k}^{*}=5$ & 0.06 & 0.384 & 0.295 & 0.035 & 0.002\tabularnewline
$\hat{k}^{*}=6$ & 0.185 & 0 & 0 & 0 & 0 & $\hat{k}^{*}=6$ & 0.005 & 0.231 & 0.414 & 0.34 & 0.138\tabularnewline
\rowstyle{\bfseries}%
$\hat{k}^{*}=k_0$ & 0.597 & 0.939 & 0.958 & 0.95 & 0.959 & $\hat{k}^{*}=k_0$ & 0 & 0.044 & 0.142 & 0.557 & 0.783\tabularnewline
$\hat{k}^{*}\geq8$ & 0.03 & 0.043 & 0.036 & 0.049 & 0.041 & $\hat{k}^{*}\geq8$ & 0 & 0.001 & 0.005 & 0.042 & 0.07\tabularnewline
\hline
\end{tabular}
\end{table}
}


\end{itemize}

\section{\bl{An example of real data analysis}}\label{application}

We analyse the log returns of 100 stocks (denoted by $y_t$), included in the S\&P500 during the period from 2005-01-03 to 2011-09-16. We have in total $T=1689$ observations with $p=100$.  Thorough eigenvalue analysis is applied to the lag-1 sample auto-covariance matrix  $\widehat{M}=\hat{\Sigma}_y{ }\hat{\Sigma}'_y{ }$ of $y_t$. The largest eigenvalue of $\widehat{M}$ is $\lambda_1(\widehat{M})=38.69$. The second to the 30th largest eigenvalues and their ratios are plotted in Fig \ref{eigratio}.

\begin{figure}[!htbp]
  \centering
  \includegraphics[width=0.8\textwidth,trim=0.0cm 1cm 1cm 0cm,clip]{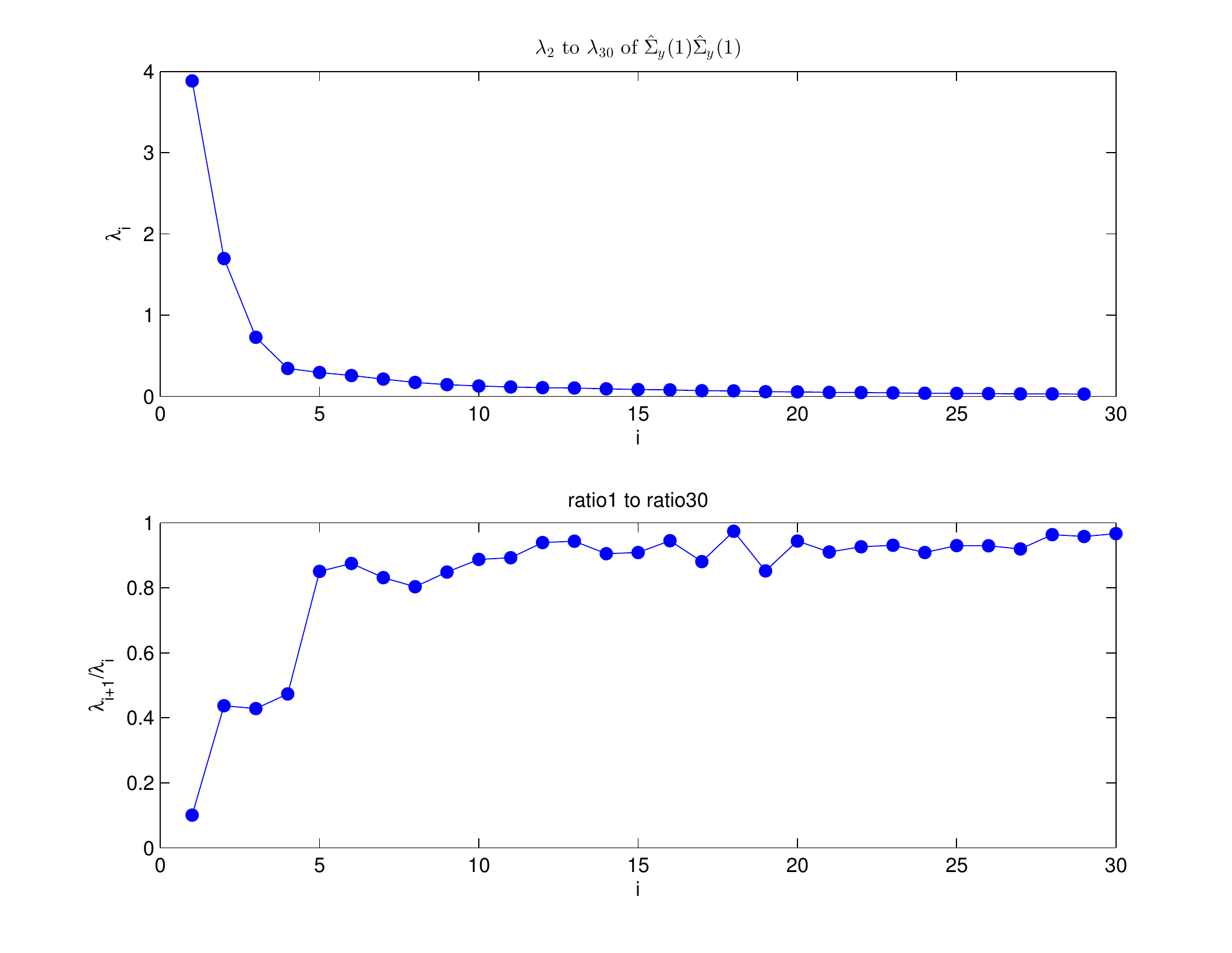}
  \caption{\small Eigenvalues of $\widehat{M}$\label{eigratio}}
\end{figure}

To estimate the number of factors, we first adopt the two-step procedure investigated by \citet{LY012} since the ratio plot in Fig \ref{eigratio} is exhibiting at least two different levels of factor strength. Obviously, in the first step,
\[\hat{r}_1=\underset{1\leq i\leq 99}{\arg\min}\lambda_{i+1}/\lambda_i=1,\]
the factor loading estimator of the first factor $\hat{A}$ is the eigenvector of $\widehat{M}$ which corresponds to the largest eigenvalue $\lambda_1$. The resulting residuals after eliminating the effect of the first factor is
\[\hat{\veps}_t=({\bf I}_{100}-\hat{A}\hat{A}')y_t.\]
Repeating the procedure in step one, we treat $\hat{\veps}$ as the original sequence $y_t$ and get the eigenvalues $\lambda_i^*s$ of the lag-1 sample auto-covariance matrix $\widehat{M}^{(1)}=\hat{\Sigma}_{\hat{\veps}}{ }\hat{\Sigma}'_{\hat{\veps}}{ }$. The 30 largest eigenvalues and their ratios are plotted in Fig \ref{eigratiostep2}.

\begin{figure}[!htbp]
  \centering
  \includegraphics[width=0.8\textwidth,trim=0.0cm 1cm 1cm 0cm,clip]{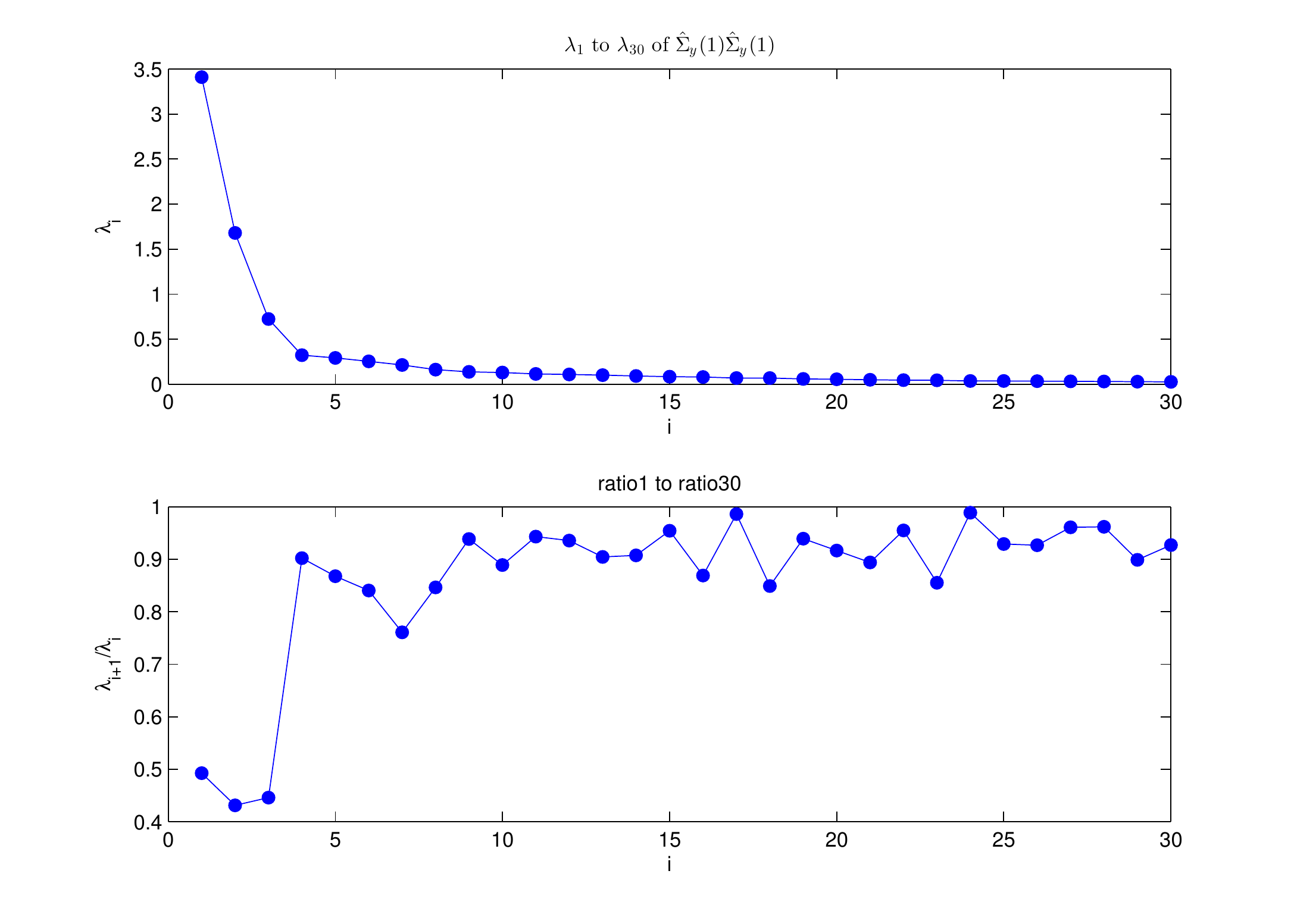}
  \caption{\small Eigenvalues of $\widehat{M}^{(1)}$ }\label{eigratiostep2}
\end{figure}

It can be seen from the second step that
\[\hat{r}_2=\underset{1\leq i\leq 99}{\arg\min}\lambda^*_{i+1}/\lambda^*_i=2,\]
the factor loading estimator of the second level factors $\hat{A}^*$ are the orthonormal eigenvectors of $\widehat{M}^{(1)}$ corresponding to the first two largest eigenvalues.

In conclusion, the two-step procedure proposed by \cite{LY012} identifies three factors in total with two different levels of factor strength. The eigenvalues of the lag-1 sample auto-covariance matrix $\widehat{M}^{(2)}$ of residuals after subtracting the three factors detected previously are shown in Fig \ref{eigratiostep3}.

\begin{figure}[!htbp]
  \centering
  \includegraphics[width=0.8\textwidth,trim=0.0cm 1cm 1cm 0cm,clip]{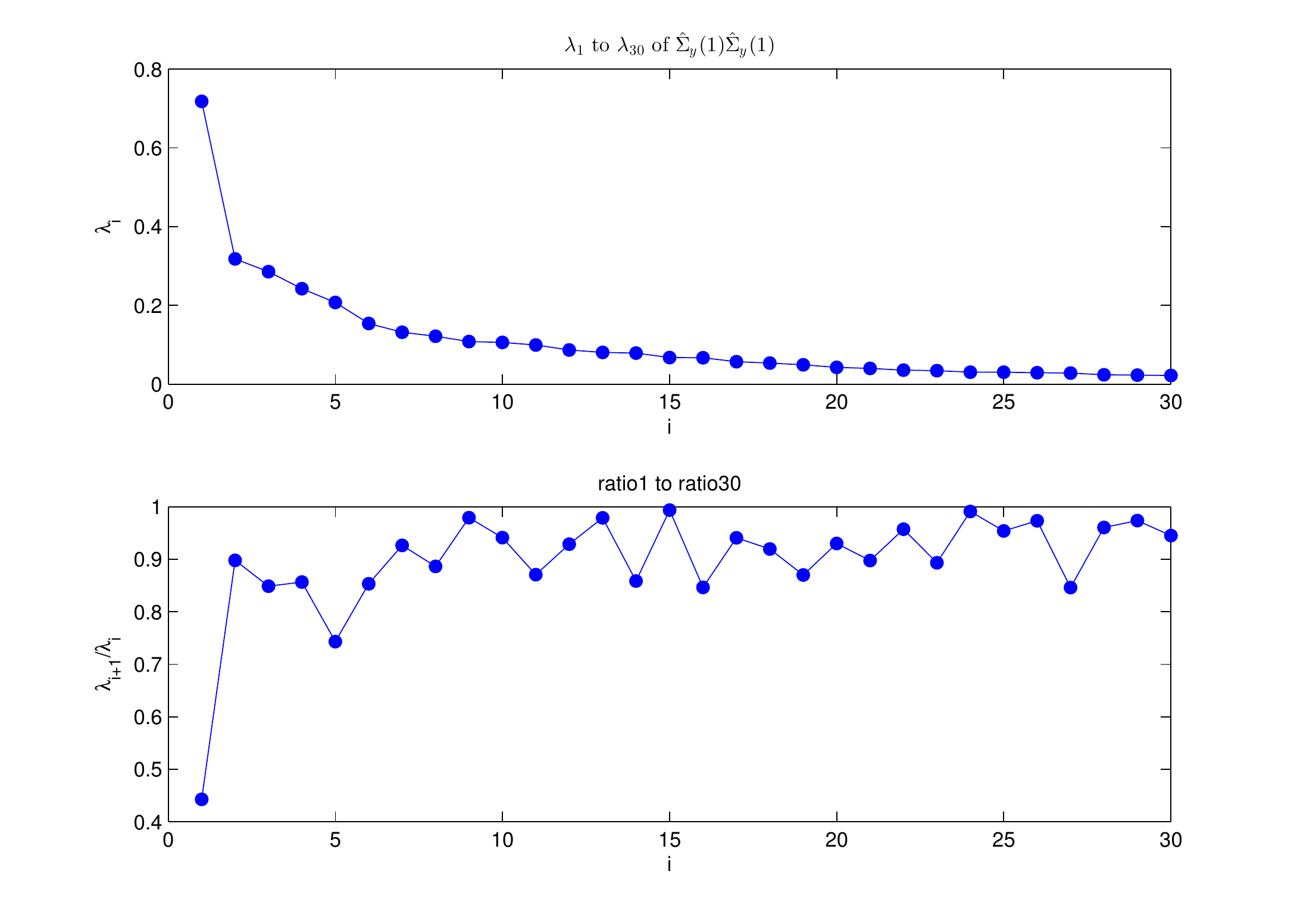}
  \caption{\small Eigenvalues of $\widehat{M}^{(2)}$ }\label{eigratiostep3}
\end{figure}

There is still one isolated eigenvalue in the eigenvalues plot. If we
go one step further and treat it as an extra factor with weakest
strength, then the eigenvalue plot of the lag-1 sample auto-covariance
matrix $\widehat{M}^{(3)}$ of residuals after eliminating four factors
looks like in Fig \ref{eigratiostep4}.

\begin{figure}[!htbp]
  \centering
  \includegraphics[width=0.8\textwidth,trim=0.0cm 1cm 1cm 0cm,clip]{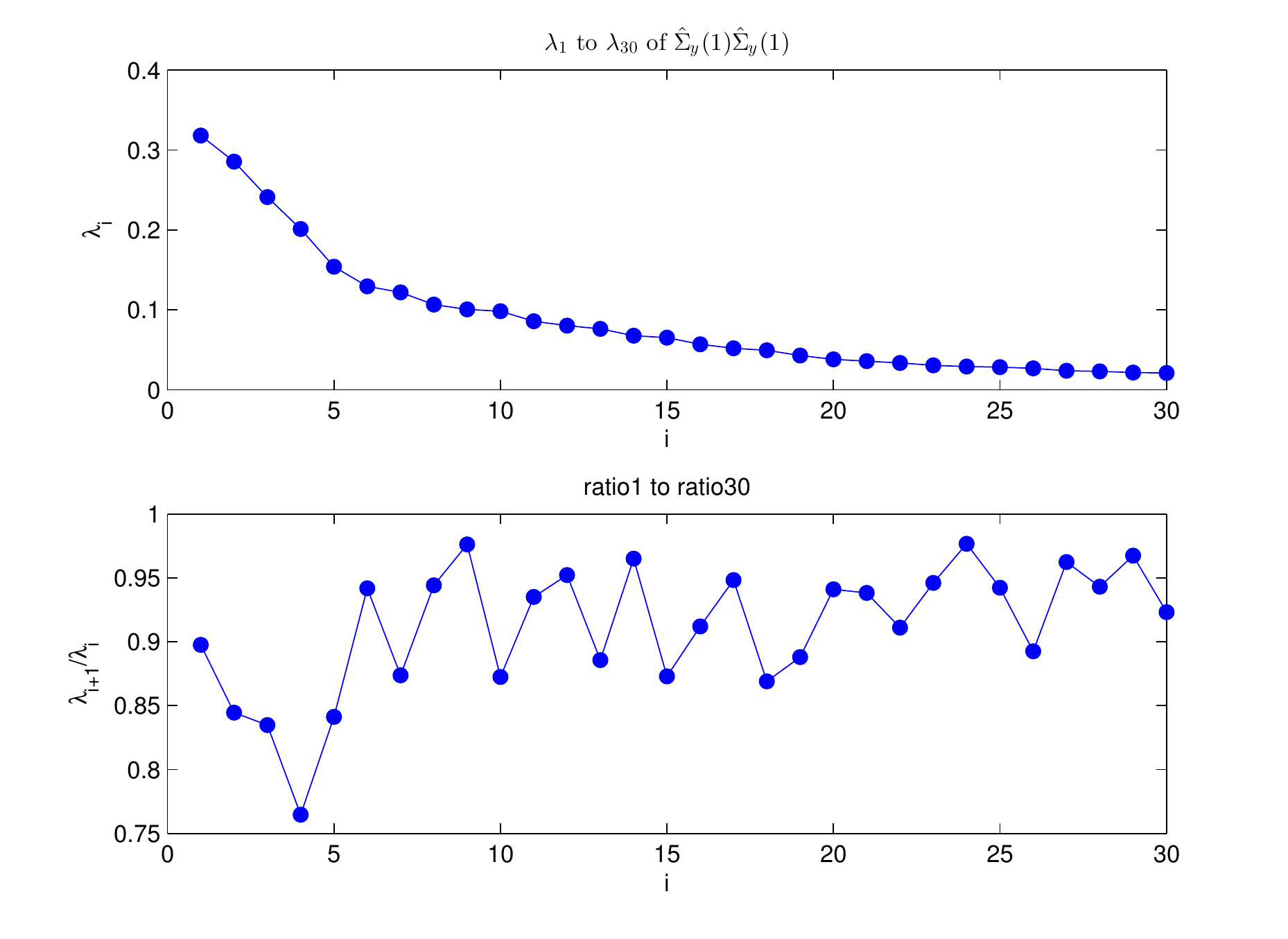}
  \caption{\small Eigenvalues of $\widehat{M}^{(3)}$ }\label{eigratiostep4}
\end{figure}

A major problem of the methodology in \cite{LY012} is that it \bl{does not}
provide a clear criterion to stop this two or multi-step
procedure. Clearly, this method can only detect factors with one level
of strength at each step and can hardly handle problems with factors
of multilevel strengths due to the lack of stopping criterion in
multi-step detection.

In the following, we use the estimator $\hat k^*$ \eqref{hattk}
of this paper to estimate the
number of factors. At first,
\bl{the tuning parameter $d_T$ is calibrated
  with $(p,T)=(100,1689)$ using the simulation method indicated in
  Section~\ref{choosedn}; the value found is $d_T=0.1713$ in this
  case.
  The eigenvalue ratios  of the sample matrix $\hat M$ are  shown
  in Figure~\ref{eigratiodn} (already displayed in  the lower panel
  of Figure~\ref{eigratio}) where the detection line of value $1-d_T=0.8287$
  is also drawn. As displayed, we found $\hat k^*=4$ factors.

In conclusion,  for this data set with $p=100$ stocks, our estimator
proposes 4 significant factors while the estimator
$\tilde k$ from  \cite{LY012} indicate
1, 3  and  4 factors when
one step, two steps  and  3 steps  are used respectively.
It appears again that
multiple steps are needed for the use of the estimator
$\hat k$ in real data analysis; it remains however unclear how to
decide the number of these necessary steps.
On the contrary, our estimator is able to identify  simultaneously all
significant factors
and the procedure is  independent of
the number of different levels of the factor  strengths.
}

\begin{figure}[!htbp]
  \centering
  \includegraphics[width=0.8\textwidth,trim=0.0cm 1cm 1cm 0cm,clip]{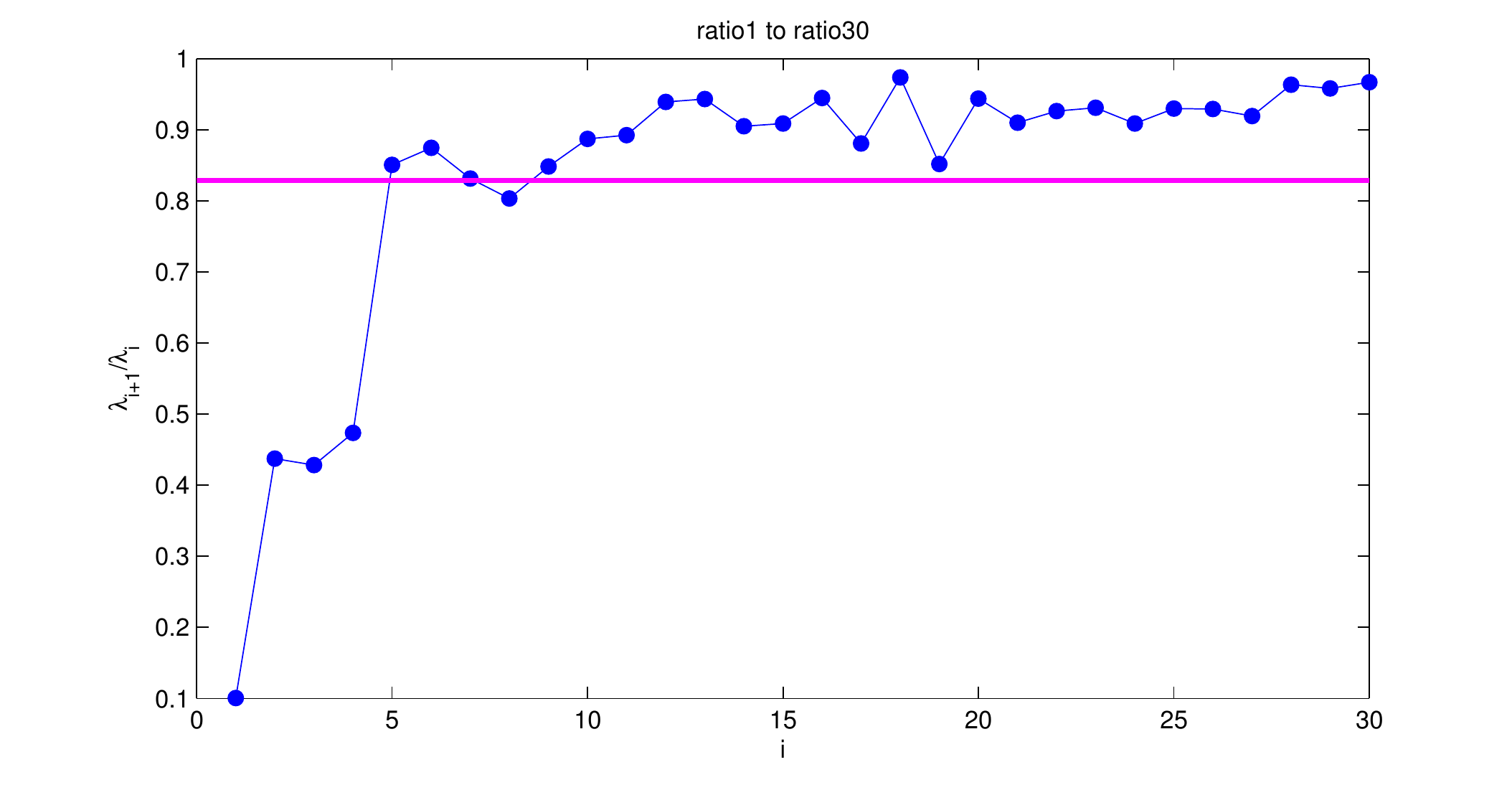}
  \caption{\small Eigenvalues of $\widehat{M}$ }\label{eigratiodn}
\end{figure}

\noindent\section*{Note}
  {A supplementary file \citet{li2}  collects several
    technical proofs used in the  paper.}

\appendix
\section{Supplementary material for the paper 
  ``Identifying the number of factors from singular values of a
  large  sample auto-covariance matrix''}


This supplementary collects several technical lemmas 
that are used  in the main paper. 

\begin{lemma}\label{l1}
  $\E(\lambda^2 I-E_1E_1'E_2E_2')^{-1}E_1E_1'$ and $\E(\lambda^2 I-E_2E_2'E_1E_1')^{-1}E_2E_2'$ are diagonal.
\end{lemma}

\begin{lemma}\label{l2}
  $
  \E(\lambda^2 I-E_1E_1'E_2E_2')^{-1}E_1E_1'E_2E_2'$ and\\ $\E(\lambda^2 I-E_2E_2'E_1E_1')^{-1}E_2E_2'E_1E_1'$ are diagonal.
\end{lemma}

\begin{proof}
  The proofs of Lemma \ref{l1} and Lemma \ref{l2} have been  already given in the paper \cite{li}.
\end{proof}

\begin{lemma}\label{33}
  \begin{align*}
    &~~~~\E\left[\lambda I_k-\lambda X_1'(\lambda^2I-E_1E_1'E_2E_2')^{-1}E_1E_1'X_1\right]\\
    &=\E\left[\lambda I_k-\lambda X_0'(\lambda^2I-E_2E_2'E_1E_1')^{-1}E_2E_2'X_0\right]\\
    &=\begin{pmatrix}
      \lambda-\frac{\lambda T(\lambda^2)}{y+T(\lambda^2)}(1+\gamma_0(1)) & \cdots & 0 \\
      \vdots & \ddots & \vdots \\
      0 & \cdots & \lambda-\frac{\lambda T(\lambda^2)}{y+T(\lambda^2)}(1+\gamma_0(k)) \\
    \end{pmatrix}
  \end{align*}
\end{lemma}

\begin{proof}
  \begin{align*}
    &~~~~\E\left[\lambda I_k-\lambda X_1'(\lambda^2I-E_1E_1'E_2E_2')^{-1}E_1E_1'X_1\right]\\
    &=\E\left\{\E\left[\lambda I_k-\lambda X_1'(\lambda^2I-E_1E_1'E_2E_2')^{-1}E_1E_1'X_1\big| X_1\right]\right\}\\
    &=\E\left\{\lambda I_k-\lambda X_1'E\left[(\lambda^2I-E_1E_1'E_2E_2')^{-1}E_1E_1'\right]X_1\big| X_1\right\}~.
  \end{align*}
  Since
  \[
  \E\left[(\lambda^2I-E_1E_1'E_2E_2')^{-1}E_1E_1'\right]
  \]
  is a diagonal $T \times T$ matrix according to Lemma \ref{l1}, and we denote it as
  \[
  diag(d_1, \cdots, d_T)~.
  \]
  Then the $(i,i)$-th element of $\E\left\{X_1'\E\left[(\lambda^2I-E_1E_1'E_2E_2')^{-1}E_1E_1'\right]X_1\right\}$ equals to
  \begin{align}\label{p1}
    &~~\E\sum_{j=1}^T X_1'(i,j)d_{j}X_1(j,i)=\sum_{j=1}^T \E X_1^2(j,i)d_j\nonumber\\
    &=\frac 1T \sum_{j=1}^T \E(x_{i\,j+1}+\varepsilon_{i\,j+1})^2\E d_j=\frac 1T \sum_{j=1}^T (1+\gamma_0(i))\E d_j\nonumber\\
    &=\frac  1T (1+\gamma_0(i))\E \tr \left[(\lambda^2I-E_1E_1'E_2E_2')^{-1}E_1E_1'\right]~,
  \end{align}
  where the second and third equalities are due to the independence between $x_t$ and $\varepsilon_t$ and the i.i.d feature of $\varepsilon_t$.
  If we denote
  \begin{align*}
    &x_T=\frac 1 T tr(E_1E_1'E_2E_2'-\lambda^2 I)^{-1}~,\\
    &y_T=\frac 1 T tr\big[(E_1E_1'E_2E_2'-\lambda^2 I)^{-1}E_1E_1'\big]~,
  \end{align*}
  there exists the relationship that:
  \begin{align*}
    1+\lambda^2 x_T=\frac{y\cdot y_T}{1+y_T}~,
  \end{align*}
  see (1) in \citet{li}. So \eqref{p1} reduces to:
  \begin{align}\label{p2}
    -(1+\gamma_0(i))\E y_T=-(1+\gamma_0(i))\E\left(\frac{1+\lambda^2x_T}{y-1-\lambda^2x_T}\right)~.
  \end{align}
  Besides,
  \begin{align*}
    \E x_T&=E\left[ \frac 1 T \tr(E_1E_1'E_2E_2'-\lambda^2 I)^{-1}\right]\\
    &=\begin{cases}
      \int_a^b \frac{1}{x-\lambda^2}f(x)dx~,~~~~y>1\\[3mm]
      \int_0^b \frac{1}{x-\lambda^2}f(x)dx-\frac{1-y}{\lambda^2}~,~~~~0<y \leq 1~,
    \end{cases}
  \end{align*}
  which leads to
  \begin{align*}
    \E (1+\lambda^2x_T)
    &=\begin{cases}
      \int_a^b \frac{x}{x-\lambda^2}f(x)dx~,~~~~y>1\\[3mm]
      \int_0^b \frac{x}{x-\lambda^2}f(x)dx~,~~~~0<y \leq 1~
    \end{cases}\\
    &=-T(\lambda^2)~,
  \end{align*}
  where $f(x)$ is the density function  of the LSD of $E_1E_1'E_2E_2'$ (also $E_2E_2'E_1E_1'$), and $T(z)$ is the $T$-transform that associated with $f(x)$ whose  support is $[a\cdotp \mathbf{1}_{\{y>1\}},b]$.

  Therefore, we have \eqref{p2} equals to
  \begin{align*}
    (1+\gamma_0(i))\frac{T(\lambda^2)}{y+T(\lambda^2)}~.
  \end{align*}

  For $i\neq k$, the $(i,k)$-th  element of \[\E\left\{X_1' \E\left[(\lambda^2I-E_1E_1'E_2E_2')^{-1}E_1E_1'\right]X_1\right\}\] equals to
  \begin{align*}
    &\E\sum_{j=1}^T X_1'(i,j)d_{j}X_1(j,k)=\sum_{j=1}^T \E X_1(j,i)X_1(j,k)\E d_j\nonumber\\
    =&\frac 1T \sum_{i=1}^T \E(x_{i\,j+1}+\varepsilon_{i\,j+1})(x_{k\,j+1}+\varepsilon_{k\,j+1})\E d_j=0,\nonumber
  \end{align*}
  due to the independence between the coordinates of $x_t$ and also the independence between $x_t$ and $\varepsilon_t$.

  All this leads to the fact that
  \begin{align*}
    ~~&~\E\left[\lambda I_k-\lambda X_1'(\lambda^2I-E_1E_1'E_2E_2')^{-1}E_1E_1'X_1\right]\\
    &=\begin{pmatrix}
      \lambda-\frac{\lambda T(\lambda^2)}{y+T(\lambda^2)}(1+\gamma_0(1)) & \cdots & 0 \\
      \vdots & \ddots & \vdots \\
      0 & \cdots & \lambda-\frac{\lambda T(\lambda^2)}{y+T(\lambda^2)}(1+\gamma_0(k)) \\
    \end{pmatrix}~.
  \end{align*}
  The same result also holds true for
  \[
  \E\left[\lambda I_k-\lambda X_0'(\lambda^2I-E_2E_2'E_1E_1')^{-1}E_2E_2'X_0\right]~.
  \]
  The proof of the Lemma is complete.

\end{proof}

\begin{lemma}\label{34}
  \begin{align*}
    ~~&~\E\left[X_1'\left(I+(\lambda^2I-E_1E_1'E_2E_2')^{-1}E_1E_1'E_2E_2'\right)X_0\right]\\
    =&\E\left[X_0'\left(I+(\lambda^2I-E_2E_2'E_1E_1')^{-1}E_2E_2'E_1E_1'\right)X_1\right]\\
    =&\begin{pmatrix}
      (1+T(\lambda^2))\gamma_1(1) & \cdots & 0 \\
      \vdots & \ddots & \vdots \\
      0 & \cdots & (1+T(\lambda^2))\gamma_1(k)\\
    \end{pmatrix}
  \end{align*}
\end{lemma}

\begin{proof}
  \begin{align*}
    ~~&~\E\left[X_1'\left(I+(\lambda^2I-E_1E_1'E_2E_2')^{-1}E_1E_1'E_2E_2'\right)X_0\right]\\
    =&\E\left\{\E\left[X_1'\left(I+(\lambda^2I-E_1E_1'E_2E_2')^{-1}E_1E_1'E_2E_2'\right)X_0\big|X_0,X_1\right]\right\}\\
    =&\E\left\{X_1'\E \left[I+(\lambda^2I-E_1E_1'E_2E_2')^{-1}E_1E_1'E_2E_2'\right]X_0\big|X_0,X_1\right\}\\
  \end{align*}
  Since
  \[
  \E(\lambda^2I-E_1E_1'E_2E_2')^{-1}E_1E_1'E_2E_2'
  \]
  is  diagonal according to Lemma \ref{l2}, and we denote it as $\text{diag}(a_1,\cdots,a_T)$.
  Then the $(i,i)$-th element of
  \[
  \E\left[X_1'\left(I+(\lambda^2I-E_1E_1'E_2E_2')^{-1}E_1E_1'E_2E_2'\right)X_0\right]
  \]
  equals to
  \begin{align*}
    &~~\E\sum_{j=1}^TX_1'(i,j)(1+a_j)X_0(j,i)\\
    =&\sum_{j=1}^T\E X_1(j,i)X_0(j,i)\E(1+a_j)\\
    =&\frac 1T \sum_{j=1}^T \E(x_{i\,j+1}+\varepsilon_{i\,j+1})(x_{i\,j}+\varepsilon_{i\,j})\E(1+a_j)\\
    =&\frac 1T \sum_{j=1}^T \E(x_{i\,j+1}x_{i\,j})\E(1+a_j)=\frac 1T \sum_{j=1}^T \gamma_1(i) \E(1+a_j)\\
    =&\frac{\gamma_1(i)}{T}\left(T+\E\tr\left[(\lambda^2I-E_1E_1'E_2E_2')^{-1}E_1E_1'E_2E_2'\right]\right)\\
    =&\gamma_1(i)\left(1+\int \frac{x}{\lambda^2-x}dF(x)\right)=\gamma_1(i)(1+T(\lambda^2))~.
  \end{align*}

  Also for $i\neq k$, the $(i,k)$-th element of
  \[
  \E\left[X_1'\left(I+(\lambda^2I-E_1E_1'E_2E_2')^{-1}E_1E_1'E_2E_2'\right)X_0\right]
  \]
  equals to
  \begin{align*}
    ~~&\E\sum_{j=1}^TX_1'(i,j)(1+a_j)X_0(j,k)=\sum_{j=1}^T\E X_1(j,i)X_0(j,k)\E(1+a_j)\\
    =&\frac 1T \sum_{j=1}^T \E(x_{i\,j+1}+\varepsilon_{i\,j+1})(x_{k\,j}+\varepsilon_{k\,j})\E(1+a_j)=0,
  \end{align*}
  where the last equality is due to the independence between the $k$ coordinates of $x_t$ and between $x_t$ and $\varepsilon_t$.

  The same is true for
  \[
  \E\left\{\E\left[X_1'\left(I+(\lambda^2I-E_1E_1'E_2E_2')^{-1}E_1E_1'E_2E_2'\right)X_0\big|X_0,X_1\right]\right\}
  \]
  and we omit the detail.

  The proof of this Lemma is complete.
\end{proof}

\end{document}